\documentclass[a4paper,11pt]{article}
\usepackage[margin=1in]{geometry}
\usepackage{setspace}
\usepackage{amsmath,amssymb,amsthm,mathtools}
\usepackage{bm}
\usepackage{graphicx}
\usepackage{booktabs}
\usepackage[T1]{fontenc}
\usepackage{lmodern}
\usepackage[protrusion=true,expansion=true]{microtype}
\usepackage[hidelinks]{hyperref}
\usepackage{enumitem}
\usepackage{caption}
\usepackage{subcaption}
\usepackage{tikz}
\usetikzlibrary{arrows.meta,positioning,shapes.geometric,shapes.symbols,calc,fit,intersections}
\usepackage{pgfplots}
\pgfplotsset{compat=1.18}
\usepackage{algorithm,algpseudocode}
\usepackage{thmtools}
\usepackage{dsfont}
\usepackage{cite}
\usepackage[protrusion=true,expansion=true]{microtype}
\usepackage[hidelinks]{hyperref}

\newcommand{\Eqref}[1]{Equation~(\ref{#1})}

\newtheorem{theorem}{Theorem}
\newtheorem{corollary}{Corollary}
\newtheorem{lemma}{Lemma}
\newtheorem{conjecture}{Conjecture}
\newtheorem{remark}{Remark}
\newtheorem{proposition}{Proposition}

\newcommand{\keywords}[1]{\vspace{0.5em}\noindent\textbf{Keywords:} \textit{#1}}

\title{Engineering Social Optimality via Utility Shaping in Non-Cooperative Games under Incomplete Information and Imperfect Monitoring}

\author{David Smith\thanks{CSIRO's Data61, Sydney, Australia, David.Smith@data61.csiro.au}\and Jie Dong\thanks{Obsidian Security, Sydney, Australia, jdong@obsidiansecurity.com}\and Yizhou Yang\thanks{Zhongguancun Laboratory, Beijing, China, yangyz@zgclab.edu.cn}}
\date{\today}

\begin{document}
\maketitle
\vspace{-1.2cm}
\onehalfspacing
\begin{abstract}
In this paper, we study decentralized decision-making where agents optimize private objectives under incomplete information and imperfect public monitoring, in a non-cooperative setting. By shaping utilities—embedding shadow prices or Karush-Kuhn-Tucker(KKT)-aligned penalties—we make the stage game an exact-potential game whose unique equilibrium equals the (possibly constrained) social optimum. We characterize the Bayesian equilibrium as a stochastic variational inequality; strong monotonicity follows from a single-inflection compressed/stretched-exponential response combined with convex pricing. We give tracking bounds for damped-gradient and best-response-with-hysteresis updates under a noisy public index, and corresponding steady-state error. The framework accommodates discrete and continuous action sets and composes with slower discrete assignment. Deployable rules include: embed prices/penalties; publish a single public index; tune steps, damping, and dual rates for contraction. Computational experiments cover (i) a multi-tier supply chain and (ii) a non-cooperative agentic-AI compute market of bidding bots. Relative to price-only baselines, utility shaping attains near-centralized welfare, eliminates steady-state constraint/capacity violations when feasible, and accelerates convergence; with quantization, discrete equilibria track continuous ones within the mesh. The blueprint is portable to demand response, cloud/edge scheduling, and transportation pricing and biosecurity/agriculture. Overall, utility shaping plus a public index implements the constrained social optimum with stable equilibria under noise and drift—an operations-research-friendly alternative to heavy messaging or full mechanism design.
\end{abstract}


\section{Introduction}
Many operations research (OR) systems feature strategic agents with private information, noisy feedback, and tight constraints: e.g., demand response in power systems, cloud/edge resource allocation, transportation networks, agentic AI and multi-tier supply chains \cite{PalenskyDietrich2011,AlbadiElSaadany2008,SimchiLevi2007,Beckmann1956}. Rich message exchange is often costly or privacy-sensitive, yet even coarse coordination can produce large gains (cost, reliability/SLA, and compliance) \cite{Cachon2003,McMahan2017}. Distributed optimization and control provide scalable coordination primitives under such constraints \cite{NedicOzdaglar2010,BoydADMM2011}, and game-theoretic control connects these primitives to equilibrium notions with robustness guarantees \cite{BasarOlsder1999}. Our principal question is: \emph{can we achieve social optimality with a unique, globally stable equilibrium using only local measurements and a low-bandwidth public index?}

Our answer is a \emph{message-free utility-shaping} blueprint. We embed shadow prices or KKT-aligned penalties into private payoffs so that selfish first-order conditions replicate those of the planner. Under mild curvature—captured by compressed/stretched exponential response surfaces that arise in practice—the induced game is an exact-potential game with a strongly monotone pseudo-gradient. The Bayesian equilibrium then solves a stochastic variational inequality (SVI), enabling clean existence/uniqueness and algorithmic tracking guarantees under noise/drift. In contrast to mechanism design with transfers or heavy signaling, our implementation requires only a \emph{single public index} (scarcity/damage/reliability) and local measurements.

Our contributions are as following:
\begin{enumerate}[label=\textit{\roman*}]
\item \textbf{Planner program $\leftrightarrow$ shaped game.}
      We formalize a welfare program with reliability/throughput proxies and show that the shaped non-cooperative stage game is an \emph{exact potential} game whose (pure) Nash equilibirum NE \emph{implements the planner’s social optimum}, including constraints via KKT-aligned prices/penalties (Eq.~(4) Dong-Smith-Hanlen (DSH)~\cite{DongSmithHanlen2016}; Eq.~(5) Yang–Smith~\cite{YangSmith2017ICC}).

\item \textbf{Curvature certificate and uniqueness}
      Using single-inflection compressed/stretched-exponential response (IEEE~802.15.8 archetype) with convex pricing, we obtain strong monotonicity/strong concavity on compact domains, yielding \emph{existence and uniqueness} of equilibrium and enabling contraction.

\item \textbf{Bayesian equilibrium as an SVI}
      Under incomplete information and a public index, equilibrium solves a \emph{stochastic variational inequality}; Lipschitz/curvature bounds deliver a \emph{unique} Bayesian NE and connect directly to first-order algorithms.

\item \textbf{Dynamics, tracking, and discrete robustness}
      We derive explicit tracking bounds that grow linearly with drift and noise and tighten as the update rule becomes more contractive (i.e., the one-step shrinkage is stronger), and we show the same guarantees hold when discrete actions are composed with continuous control.

\item \textbf{Two domains, consistent gains}
      Computational studies in (i) multi-tier supply chains and (ii) a non-cooperative agentic-AI compute market show near-centralized welfare, markedly fewer violations (vanishing under feasible capacity with dual damping), and faster convergence than price-only baselines.

\item \textbf{Deployable playbook \& limits}
      We give a one-screen recipe: choose an interpretable public index; fit the two parameters that shape the sigmoid curve relating “effort” to reliability/throughput; and tune (a) the agents’ update step and damping, (b) the index’s adjustment speed, and (c) a small hysteresis band so the overall loop is contractive. We also specify scope/limitations (flat or multi-inflection curves, strong complementarities, delayed/noisy indices) and practical guardrails (more damping/smoothing, event triggers, discrete-convex composition) with brief managerial implications.

\end{enumerate}
For summarized benefits of our contribution, please see Table 1 here

% ===== Table: Baseline vs shaped-utility outcomes =====
\begin{table}[h]
\centering
\caption{Baseline vs. shaped-utility outcomes (qualitative summary).}
\label{tab:baseline-vs-shaped}
\begin{tabular}{lcc}
\toprule
Outcome & Price-only baseline & Utility shaping (+ public index) \\
\midrule
Welfare vs planner optimum & lower / variable & near-optimal (tight) \\
Constraint violations (steady-state) & possible/persistent & eliminated when feasible \\
Convergence speed & slower / oscillatory & accelerated, contractive \\
Robustness to noise/drift & limited & explicit tracking bounds \\
Discrete + continuous composition & ad hoc & well-posed under shared potential \\
Messaging overhead & low & unchanged (no runtime messaging) \\
\bottomrule
\end{tabular}
\end{table}
\begin{flushleft} Please find here then, Fig. \ref{fig:public-index-loop} explaining the use of message blueprint as message-free loop\end{flushleft}
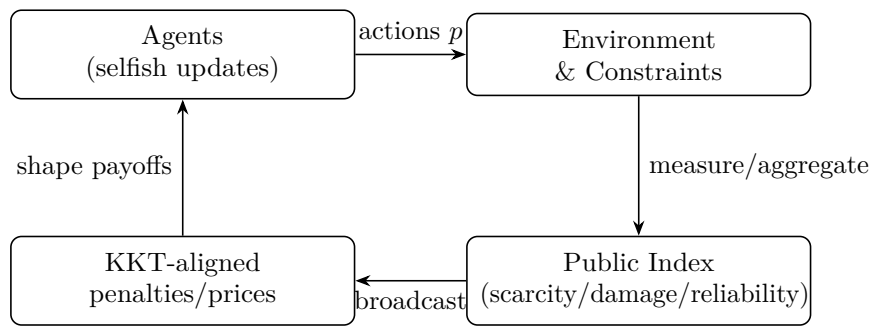
\begin{figure}[h]
\centering
\begin{tikzpicture}[>=Stealth, line width=0.6pt, font=\small,
  box/.style={draw, rounded corners, align=center, text width=4.1cm, minimum height=1.05cm, inner sep=6pt}]
  % --- Fixed positions (no overlap) ---
  \node[box] (agents) at (0,0) {Agents\\(selfish updates)};
  \node[box] (env)    at (6.0,0) {Environment\\\& Constraints};
  \node[box] (index)  at (6.0,-3.0) {Public Index\\(scarcity/damage/reliability)};
  \node[box] (pen)    at (0,-3.0) {KKT-aligned\\penalties/prices};

  % --- Arrows ---
  \draw[->] (agents) -- node[midway, above] {actions $p$} (env);
  \draw[->] (env)    -- node[midway, right] {measure/aggregate} (index);
  \draw[->] (index)  -- node[midway, below] {broadcast} (pen);
  \draw[->] (pen)    -- node[midway, left] {shape payoffs} (agents);
\end{tikzpicture}
\caption{Message-free loop: a single public index focalizes decentralized responses; private payoffs are shaped with KKT-aligned penalties so selfish updates implement the planner’s first-order conditions.}
\label{fig:public-index-loop}
\end{figure}

\paragraph{Organization.}
This paper is organized as follows. Section~\ref{sec:literature} reviews related work. Section~\ref{sec:problem} formulates the planner’s welfare program and the shaped non-cooperative stage game. Section~\ref{sec:methodology} develops the methodology for engineering social optimality. Section~\ref{sec:experiments} presents two studies: a multi-tier supply-chain study, and a non-cooperative agentic-AI compute-market. Section~\ref{sec:design-rules} distills deployable design rules and reporting guidelines. Section~\ref{sec:managerial} provides insights for practitioners. Section~\ref{sec:scope} discusses scope and limitations. Section~\ref{sec:conclusion} provides some concluding remarks.

\section{Related Work}
\label{sec:literature}

\subsection{Pricing, decentralization, and proportional fairness}
Early work established how shadow prices decentralize resource allocation while preserving
stability and fairness in networks \cite{Kelly1998,LowLapsley1999,ChiangProcIEEE2007,JohariTsitsiklis2004}.
These ideas underpin our DSH-style alignment for Eq.~(4): prices embedded in local utilities let
selfish first-order conditions reproduce the planner’s gradient, avoiding rich message exchange.

\subsection{Power control, response curves, and monotone games}
Uplink power control and interference-coupled systems admit monotone structures with convergence
guarantees \cite{Yates1995}. In short-range wireless, the empirical SINR$\!\to$PDR map is well captured
by compressed/stretched-exponential sigmoids in IEEE~802.15.8 \cite{IEEE802158-2017}. Potential and
stable game concepts \cite{MondererShapley1996,HofbauerSandholm2009,Scutari2014} explain why curvature
and diagonal dominance yield uniqueness and contraction---insights we transfer to OR settings and
make explicit via our curvature certificate (Sec.~4.2).

\subsection{Variational inequalities, stochastic programs, and discrete convexity}
Variational inequalities and monotone operator theory provide existence/uniqueness and algorithmic
foundations for equilibria and fixed points \cite{FacchineiPang2003,RockafellarWets1998,ShapiroDentchevaRuszczynski2014,CombettesPesquet2011,JuditskyNemirovski2011}.
We use these tools to formalize Bayesian equilibria as SVIs and to derive tracking bounds. For
discrete/quantized actions and ordered strategy sets, monotone selection and supermodularity
\cite{Tarski1955,Topkis1998,Topkis1979} justify our discrete robustness claims and two-layer composition.

% =============================================================
% 3. PROBLEM DESCRIPTION (planner program + game link)
% =============================================================
\section{Problem description}
\label{sec:problem}

\paragraph*{Symbols and definitions.}
\begin{description}[leftmargin=1em, itemsep=0.05pt]
  \item[$N$] number of agents; $i\in\{1,\dots,N\}$ indexes agents.
  \item[$p_i\in X_i$] agent $i$'s action/decision (agent's control parameter e.g., power, production throughput, flow); $X=\prod_i X_i$.
  \item[$W(p)$] planner (social) welfare to be maximized; $g(p)\le 0$ denotes system constraints; $C$ a capacity/limit when used.
  \item[$z$] \emph{public index} (scarcity/damage/constraint gap) broadcast to all agents.
  \item[$\mathrm{Reli}(\cdot)$] reliability/throughput proxy formed from a signal map and a sigmoid curve.
  \item[$\gamma$] Signal-to-interference+noise-ratio (e.g., SINR); \quad $\mathrm{PDR}(\gamma)$ packet-delivery ratio (if used).
  \item[$\kappa\!>\!0,\ \beta\!>\!0$] Sigmoid curve parameters in $y(x)=\exp\!\big(-( \kappa/x)^{\beta}\big)$ (location/steepness).
  \item[$u_i,\ \tilde u_i$] private utility and shaped utility (price/penalty added to align with planner KKT).
  \item[$\lambda$] shadow price (or penalty multiplier) used in shaping; “KKT-aligned” means gradients match the planner’s first-order conditions at the solution.
  \item[$F$] pseudo-gradient/operator stacking players’ gradients (sign convention as in the paper).
  \item[$\mu\!>\!0,\ L\!>\!0$] strong-monotonicity and Lipschitz moduli of the operator on $X$.
  \item[$\eta$] agent step size; \quad $\rho$ damping (relaxation) factor; \quad $\eta_z$ index (dual) step; \quad $h$ hysteresis band for discrete actions.
  \item[$\alpha\!\in(0,1)$] resulting contraction factor of the one-step update (value depends on $\eta,\rho,\mu,L$).
  \item[$\sigma$] effective gradient-noise level; “drift” denotes slow change of the equilibrium target.
\end{description}

Players $i\in\{1,\dots,N\}$ choose controls $p_i\in X_i$; welfare is
\begin{equation}
W(p)=\sum_i v_i\!\big(\mathrm{Rel}_i(p)\big)-\lambda\sum_i c_i(p_i),\qquad X=\prod_i X_i,
\label{eq:W}
\end{equation}
with $v_i$ concave increasing, $c_i$ convex, and $\mathrm{Rel}_i$ a reliability/throughput proxy (e.g., PDR). The planner’s constrained optimum solves
\begin{equation}
\label{eq:SO}
p^{\mathrm{SO}}\in\arg\max_{p\in X}\; W(p)\quad \text{s.t.}\; g(p)\le 0,\; h(p)=0.
\end{equation}
\emph{Utility shaping} embeds shadow prices or KKT penalties in private payoffs so that $\nabla_{p_i}u_i=\nabla_{p_i}W$ (unconstrained) or reproduces the KKT stationarity of \Eqref{eq:SO}. The shaped stage game is concave with an exact potential $\Phi\equiv W$, so its unique Nash equilibrium coincides with $p^{\mathrm{SO}}$.

\section{Methodology}
\label{sec:methodology}
\subsection{Physics-informed response and utility shaping}%\label{sec:background}
IEEE 802.15.8 \cite{IEEE802158-2017} is the Peer Aware Communications standard for device-to-device proximity networking in wireless personal area networks (WPAN), where this simple, decentralized, message-light control keeps links reliably operating in dense, mobile settings. Clause 14 therein by Smith implements a thermostat-like transmit-power loop (non-cooperative game): each device measures SINR (Signal-to-Interference+Noise Ratio), uses an S-shaped SINR$\rightarrow$PDR (Signal-to-Interference+Noise) compressed exponential sigmoidal curve to predict reliability, and steps power up or down from discrete levels to meet a target PDR while saving energy and limiting interference, with PDR curve as follows:
\begin{equation}
\mathrm{PDR}(\gamma)=\exp\!\left(-\bigl(\kappa/\gamma\bigr)^{\beta}\right),\quad \kappa>0,\;\beta>1\footnote{Extended to $\beta>0$ here},
\label{eq:pdr}
\end{equation}
with unique inflection at $\gamma^{\star}=\kappa(\beta/(\beta+1))^{1/\beta}$ and $\mathrm{PDR}(\gamma^{\star})=e^{-1-1/\beta}$. Representative utilities were later developed as
\begin{align}
\label{eq:utility1}
 u_i(p) &= v_i\!\left(\mathrm{PDR}_i(\gamma_i(p))\right)-\lambda\,c_i(p_i),\\[-1mm]
\label{eq:utility2}
 U_i(p_i) &= -C\,p_i^{\,w}+\log\!\bigl(1+\mathrm{PDR}_i(\gamma_i(p))^{\,v}\bigr),\quad C,w,v>0,
\end{align}
~\cite{DongSmithHanlen2016,YangSmith2017ICC} respectively

\paragraph{Information about these prior mechanisms.}
Equation~(4) instantiates the \emph{Dong--Smith--Hanlen} (DSH) utility-shaping mechanism: a single scalar price $\lambda$ (or KKT-aligned penalty) is embedded in private payoffs so that each agent’s selfish first-order condition aligns with the planner’s gradient, i.e., $\nabla_{p_i}u_i=\nabla_{p_i}W$. The induced stage game is an \emph{exact potential} game with potential
$\Phi(p)=\sum_i v_i(\mathrm{PDR}_i(\gamma_i(p)))-\lambda\sum_i c_i(p_i)\equiv W(p)$,
so any Nash equilibrium maximizes $W$ and—under the compressed/stretched-exponential SINR$\!\to$PDR curvature— is \emph{unique} (prices-not-messages design; cf.\ \cite{DongSmithHanlen2016}).

Equation~(5) matches the \emph{Yang--Smith} augmentation: the explicit energy term $-C p_i^{\,w}$ and the \emph{log-saturating} reliability benefit $\log(1+\mathrm{PDR}_i(\gamma_i(p))^{\,v})$ preserve the same potential while \emph{strengthening curvature}. This yields strict concavity in each agent’s decision, robustness to \emph{discrete} action sets, and contraction-safe best-response/gradient updates under imperfect public monitoring (cf.\ \cite{YangSmith2017ICC}). In both cases, run-time coordination reduces to broadcasting a \emph{single public index} (a scarcity/constraint proxy) that agents combine with local measurements; rich message exchange is unnecessary while the centralized KKT system is implemented by selfish play.

\begin{theorem}[Exact potential \& uniqueness for Eq.~(4)]
\label{thm:pot-unique-eq4}
Let $X=\prod_{i=1}^N X_i$ be convex and compact. Consider Eq.~(4)
\[
u_i(p)\;=\;v_i\!\big(\mathrm{PDR}_i(\gamma_i(p))\big)\;-\;\lambda\,c_i(p_i),
\]
with $v_i$ concave and increasing, $c_i$ convex, and define the welfare/potential
\begin{equation}
\Phi(p)\;\equiv\;W(p)\;=\;\sum_{i=1}^N v_i\!\big(\mathrm{PDR}_i(\gamma_i(p))\big)\;-\;\lambda\sum_{i=1}^N c_i(p_i).
\end{equation}
Then:
\begin{enumerate}[label=(\alph*),leftmargin=1.4em]
\item The game is an exact potential game: for all $i$ and fixed $p_{-i}$,
      $\Phi(p_i,p_{-i})-\Phi(q_i,p_{-i})=u_i(p_i,p_{-i})-u_i(q_i,p_{-i})$.
\item If $W$ is \emph{strictly concave} on $X$ (e.g., because at least one of: (i) each $v_i$ is strictly concave on the relevant PDR range; (ii) the interference map $\gamma(\cdot)$ together with the compressed/stretched-exponential curvature yields a negative-definite symmetric part of the pseudo-Jacobian on $X$; (iii) $c_i$ are \emph{strongly} convex), then $W$ has a unique maximizer $p^\star$, and the stage game admits a \emph{unique} (pure) Nash equilibrium at $p^\star$.
\end{enumerate}
\end{theorem}

\begin{proof}
(a) By construction,
\begin{align}
\Phi(p_i,p_{-i})-\Phi(q_i,p_{-i})
&= v_i(\mathrm{PDR}_i(\gamma_i(p_i,p_{-i})))-v_i(\mathrm{PDR}_i(\gamma_i(q_i,p_{-i})))
-\lambda\big(c_i(p_i)-c_i(q_i)\big)\\
&= u_i(p_i,p_{-i})-u_i(q_i,p_{-i}),
\end{align}
hence $\nabla_{p_i}\Phi=\nabla_{p_i}u_i$ whenever gradients exist; the game is an exact potential game.

(b) If $W$ is strictly concave on convex compact $X$, it admits a unique maximizer $p^\star$. In an exact potential game, any (pure) NE is a (local) maximizer of $\Phi\equiv W$, and conversely any maximizer of $W$ is a (pure) NE because unilateral deviations cannot increase $W$. Strict concavity rules out multiple maximizers, hence the NE is unique and equals $p^\star$.
\end{proof}

\medskip

\begin{proposition}[Strict concavity and discrete robustness for Eq.~(5)]
\label{prop:discrete-robust-eq5}
Let each $X_i$ be either an interval $[l_i,u_i]$ or a finite $\Delta$-quantized subset of an interval. Consider Eq.~(5)
\[
U_i(p_i,p_{-i})\;=\;-\;C\,p_i^{\,w}\;+\;\log\!\Big(1+\mathrm{PDR}_i(\gamma_i(p))^{\,v}\Big),\qquad C>0,\;w>1,\;v>0.
\]
Assume the SINR$\!\to$PDR map follows the compressed/stretched-exponential geometry (single inflection, bounded slope/curvature on the operating range) and the standard interference model $\gamma(\cdot)$ is smooth on $X$.
Then:
\begin{enumerate}[label=(\alph*),leftmargin=1.4em]
\item (\emph{Strict concavity in the continuous case}) For fixed $p_{-i}$, $U_i(\cdot,p_{-i})$ is \emph{strictly concave} on $X_i$; hence each player’s best response is single-valued and continuous, and the continuous game admits a unique Nash equilibrium $\bar p$ (by Rosen’s diagonally strict concavity or, equivalently, strong monotonicity of the pseudo-gradient).
\item (\emph{Discrete robustness by quantization}) Suppose $X_i^\Delta$ is obtained by quantizing $X_i$ with mesh size at most $\Delta$ and define the nearest-neighbor quantizer $Q_\Delta(\cdot)$ componentwise. Then the discrete game admits a (pure) NE $p^\Delta$, and for all sufficiently small $\Delta$,
\[
\|p^\Delta-\bar p\|_\infty \le \Delta,
\]
i.e., $p^\Delta$ is the componentwise nearest neighbor of $\bar p$ whenever no coordinate of $\bar p$ lies exactly at a midpoint (otherwise, a deterministic tie-breaking rule selects one of the two neighbors). In particular, $p^\Delta$ is unique under fixed (deterministic) tie-breaking.
\end{enumerate}
\end{proposition}

\begin{proof}
(a) For fixed $p_{-i}$, the term $-C p_i^{\,w}$ is strictly concave on $[l_i,u_i]$ when $w>1$. The map
$g_i(p_i)=\log\!\big(1+\mathrm{PDR}_i(\gamma_i(p_i,p_{-i}))^{\,v}\big)$ is concave in $p_i$ because:
(i) $x\mapsto \log(1+x^v)$ is increasing and concave for $x\in(0,1)$ and $v>0$; (ii) $\mathrm{PDR}_i(\cdot)$ is increasing with bounded curvature on the operating range (single inflection, Proposition on exponential geometry); (iii) under standard interference models, $\gamma_i$ is (weakly) concave in own action along the feasible interval. Composition preserves concavity here, and the sum of a strictly concave term with a concave term is strictly concave. Therefore, each player’s best response is unique and continuous; diagonal strict concavity (negative definite symmetric part of the pseudo-Jacobian) or the strong-monotonicity certificate implied by the curvature yields a unique continuous NE $\bar p$.

(b) Existence: finite games admit (pure) NE in exact/ordinal potential settings; in our case, strict single-peaked best responses over a linearly ordered finite set ensure existence via the standard Tarski/Topkis monotone selection argument~\cite{Tarski1955,Topkis1979}. Proximity: by strict concavity, each continuous best response maximizer is unique; quantization to $X_i^\Delta$ selects the nearest admissible point in $X_i^\Delta$ (or one of the two nearest under midpoint ties). At a fixed point of these quantized best responses, $p^\Delta=Q_\Delta(\bar p)$ componentwise whenever $\bar p$ is not exactly on a midpoint; hence $\|p^\Delta-\bar p\|_\infty\le \Delta$. Uniqueness follows once a deterministic tie-breaking rule is fixed.
\end{proof}

\medskip

\begin{remark}[Operational takeaway: implementing the (possibly constrained) social optimum, under the stated curvature/monotonicity and convexity assumptions]
\label{rem:operational-takeaway}
\emph{What the mechanisms guarantee.}
Eq.~(4) (DSH) aligns each agent’s selfish first-order condition with the planner’s \emph{welfare gradient}, so the unique Nash equilibrium (when $W$ is strictly concave/strongly concave) \textbf{equals the planner’s social optimum} for the unconstrained program $\max_{p\in X} W(p)$.
When system constraints matter (capacities, reliability thresholds, budgets), the same blueprint—with a \emph{public index} that tracks the relevant dual variable(s) or \emph{KKT-aligned penalties} embedded in utilities—implements the \textbf{constrained social optimum}: at equilibrium, primal feasibility holds and complementary slackness is satisfied (the index plays the role of the shadow price).

Eq.~(5) (Yang--Smith) strengthens curvature via the energy term $-C p_i^{\,w}$ and the log-saturating reliability term, preserving the potential but making the pseudo-gradient \emph{strongly} monotone. This yields: (i) a \textbf{unique} equilibrium that \textbf{coincides with the planner’s constrained optimum} (the KKT point) under the index/penalty construction; (ii) \textbf{contraction-safe} dynamics (robust tracking under noise/drift); and (iii) \textbf{discrete robustness}: with quantized actions, the discrete equilibrium tracks the continuous optimum within the quantization step (\S\ref{prop:discrete-robust-eq5}), while dual/index updates enforce complementary slackness (no steady-state violations when the primal is feasible).

\emph{\textbf{How to run it in practice, generalised moving beyond WPANs}} Broadcast a single public index (scarcity/damage/constraint gap) updated by a damped dual-like rule; agents update locally using the shaped utilities in Eq.~(4) or Eq.~(5). Choose steps/damping to keep the mapping contractive; then the induced equilibrium \textbf{implements the social optimum} (unconstrained for Eq.~(4), constrained via KKT alignment for either Eq.~(4) with prices or Eq.~(5) with penalties), without rich messaging.
\end{remark}

\subsection{Curvature certificate and properties}

Thus now continuing with generalisation in practice, we investigate the required utility shaping curvature as compressed and also now stretched exponentials, how they enable a unique Nash equilibrium and exact exponential, then expand to Bayesian equilibrium. We then conjecture re the need of these particular sigmoids to obtain social optimality.
\begin{proposition}
\label{prop:exp-geometry}
For $y(x)=\exp( - (\kappa/x)^{\beta})$ with $\kappa,\beta>0$: $y\in(0,1)$, $y'\!>0$, and a unique inflection at $x_\ast=\kappa(\beta/(\beta+1))^{1/\beta}$ with $y(x_\ast)=e^{-1-1/\beta}\in(e^{-2},e^{-1})$. This curvature supplies diagonal dominance in the Jacobian of the pseudo-gradient and under convex pricing implies strong monotonicity on compact $X$.
\end{proposition}

\begin{proof}
See Appendix~\ref{app:proofs} (Prop.~\ref{prop:exp-geometry-proof}) for details.
\end{proof}

\begin{proposition}[Sigmoidal certificate $\Rightarrow$ exact potential \& unique NE]
\label{prop:min-sigmoid-sufficient}
Let $y(x)=\exp\!\bigl(-(\kappa/x)^{\beta}\bigr)$ with $\kappa>0$ and $\beta>0$ (stretched if $0<\beta<1$, compressed if $\beta>1$).
Consider utilities of the form
\[
u_i(p)\;=\;v_i\!\bigl(\mathrm{Rel}_i(p)\bigr)\;-\;\lambda\,c_i(p_i),\qquad
\mathrm{Rel}_i(\cdot)\ \text{obtained by composing an SINR/``effective signal'' map with }y(\cdot),
\]
where each $v_i$ is increasing and concave, each $c_i$ is convex, and $X=\prod_i X_i$ is convex and compact.
Define the welfare/potential $\Phi(p)\equiv W(p)=\sum_i v_i(\mathrm{Rel}_i(p))-\lambda\sum_i c_i(p_i)$. Then:
\begin{enumerate}[label=(\roman*), leftmargin=1.4em]
\item The induced stage game is an \emph{exact potential} game with potential $\Phi$, i.e., $\nabla_{p_i}\Phi=\nabla_{p_i}u_i$ for all $i$ (whenever gradients exist). Hence any (pure) Nash equilibrium maximizes $W$ over $X$, and any maximizer of $W$ is a (pure) NE.
\item If, on $X$, the SINR/“effective signal” map is smooth and the composition with $y$ yields bounded slope/curvature, then (together with convex pricing) the pseudo-gradient $F(p)=(-\nabla_{p_i}u_i)_i$ is \emph{strongly monotone} (equivalently, $W$ is \emph{strongly concave}). Therefore the maximizer of $W$ is \emph{unique}, and the stage game admits a \emph{unique} (pure) NE that coincides with the planner’s maximizer.
\end{enumerate}
\end{proposition}

\begin{proof}
For (i), by construction
\[
\Phi(p_i,p_{-i})-\Phi(q_i,p_{-i}) \;=\; v_i\!\big(\mathrm{Rel}_i(p_i,p_{-i})\big)-v_i\!\big(\mathrm{Rel}_i(q_i,p_{-i})\big)\;-\;\lambda\big(c_i(p_i)-c_i(q_i)\big)
\;=\; u_i(p_i,p_{-i})-u_i(q_i,p_{-i}),
\]
so $\nabla_{p_i}\Phi=\nabla_{p_i}u_i$ and the game is an exact potential game. Hence (pure) NEs are precisely the maximizers of $W$ on $X$ (and conversely).

For (ii), write $\mathrm{Rel}_i=\tilde y\circ s_i$ where $s_i$ is the SINR/“effective signal” map and $\tilde y(x)=y(x)$ on the operating range. Since $y(x)=e^{-(\kappa/x)^\beta}$ is strictly increasing with a \emph{single} inflection and has bounded first/second derivatives on any compact signal range, the chain rule gives bounded Jacobians/Hessians for $v_i\circ \mathrm{Rel}_i$ on $X$ (because $v_i$ is concave increasing). Together with convex $c_i$ and standard interference smoothness, the symmetric part of the pseudo-Jacobian of $F$ is negative definite (diagonal dominance/curvature transfer), which is equivalent to strong monotonicity of $F$ and strong concavity of $W$. A strongly concave $W$ has a unique maximizer on $X$; by (i) this unique maximizer is the unique (pure) NE.
\end{proof}

\begin{conjecture}[Minimality/necessity for social optimality]
\label{conj:min-sigmoid-necessary}
In incomplete-information, imperfect-monitoring settings covered here, achieving (possibly constrained) \emph{social optimality via utility shaping} with a \emph{unique}, globally stable equilibrium requires a single-inflection, compressed/stretched-exponential–type response surface (of the form $y(x)=\exp(-(\kappa/x)^\beta)$ up to monotone reparameterisations).
That is, without such curvature (or an equivalently strong monotonicity certificate), exact-potential alignment alone may not ensure uniqueness and contraction, and the shaped non-cooperative play may fail to implement the social optimum.
\end{conjecture}

% ===== Curvature & geometry of the compressed/stretched exponential =====

% ===== Potential/KKT alignment & welfare implementation =====
\begin{proposition}[Exact potential and KKT alignment]
\label{prop:potential-kkt}
Let $X=\prod_{i=1}^N X_i$ be convex compact and
\[
u_i(p)\;=\;v_i\!\big(\mathrm{Rel}_i(p)\big)-\lambda\,c_i(p_i),
\qquad
\Phi(p)\;=\;\sum_{i=1}^N v_i\!\big(\mathrm{Rel}_i(p)\big)-\lambda\sum_{i=1}^N c_i(p_i)\equiv W(p),
\]
with $v_i$ concave increasing and $c_i$ convex.
Then the game is an \emph{exact potential game} with potential $\Phi$ and every (pure) Nash equilibrium $p^\star$
maximizes $W$ over $X$. If the planner’s problem has convex constraints, replacing $u_i$ by a KKT-aligned
penalty version yields selfish first-order conditions that reproduce the centralized KKT system.
\end{proposition}

\begin{proof}
Exact potential property: for any $i$ and any $p_{-i}$, the difference
$\Phi(p_i,p_{-i})-\Phi(q_i,p_{-i})=u_i(p_i,p_{-i})-u_i(q_i,p_{-i})$ by construction, hence
$\nabla_{p_i}\Phi=\nabla_{p_i}u_i$ whenever gradients exist.
Thus a (pure) NE $p^\star$ satisfies $\nabla_{p_i}u_i(p^\star)=0$ with $p^\star\in X$, i.e.,
the KKT stationarity for $\max_{p\in X}\Phi(p)$,
which shows $p^\star$ maximizes $W\equiv\Phi$ over $X$.
Conversely, any maximizer of $\Phi$ must satisfy $\nabla_{p_i}\Phi=0$ in unconstrained
directions, i.e., $\nabla_{p_i}u_i=0$, so no player can improve—hence it is a NE.
For convex inequality constraints $g(p)\le0$, define $U_i$ by embedding penalties/prices so
$\nabla_{p_i}U_i=\nabla_{p_i}\big(W+\mu^\top g\big)$ at the optimum $(p^\star,\mu^\star)$;
then selfish FOCs coincide with the planner’s KKT stationarity.
\end{proof}

\begin{remark}[Constrained case via penalties/prices]
\label{rem:kkt}
With constraint targets (e.g., reliability thresholds), embedding KKT-aligned penalties in $U_i$ reproduces the centralized KKT stationarity in selfish first-order conditions, implementing the constrained optimum without runtime messages beyond a public index.
\end{remark}

\subsection{Bayesian equilibrium as an SVI}\label{sec:bayesian}
We now look to establishing a Bayesian equilibrium. Let private types be $\theta$ and the public index $z$. Define the pseudo-gradient $F(p;\theta)=(-\nabla_{p_i}u_i(p;\theta))_i$. Conditioned on $z$, the Bayesian equilibrium $p^{\star}$ solves
\begin{equation}
\label{eq:VI}
\mathbb{E}\![F(p^{\star};\theta)\mid z]^{\!\top}(q-p^{\star})\ge 0,\qquad \forall q\in X,
\end{equation}
and is \emph{unique} if $F_z(p):=\mathbb{E}[F(p;\theta)\mid z]$ is strongly monotone:
$ (F_z(p)-F_z(q))^{\top}(p-q)\ge \mu\|p-q\|^2$ for some $\mu>0$ \cite{FacchineiPang2003,Scutari2014}.

\begin{theorem}[Bayesian equilibrium as unique SVI solution]
\label{thm:svi-unique}
Let $F(p;\theta)=\big(-\nabla_{p_i}u_i(p;\theta)\big)_{i=1}^N$ and
$F_z(p)=\mathbb{E}[F(p;\theta)\mid z]$.
Assume:
\begin{enumerate}[label=(A\arabic*),leftmargin=1.4em]
\item $X$ is convex compact;\quad
\item $u_i(\cdot\,;\theta)$ is concave in $p_i$ for each $\theta$ and $i$;
\item $\mathrm{Rel}_i$ composed with $y$ in Prop.~\ref{prop:exp-geometry} and convex prices ensure $W$ is \emph{strongly concave}
      with modulus $\mu>0$ on $X$ (equivalently, $-\nabla W$ is $\mu$-strongly monotone);
\item $F_z$ is $L$-Lipschitz on $X$.
\end{enumerate}
Then $F_z$ is $\mu$-strongly monotone on $X$, and the Bayesian Nash equilibrium $p^\star$ is the \emph{unique} solution of
the SVI: $F_z(p^\star)^\top(q-p^\star)\ge 0$ for all $q\in X$.
\end{theorem}

\begin{proof}
From exact potential, $F_z(p)=-\nabla\Phi_z(p)$ where $\Phi_z(p)=\mathbb{E}[\Phi(p;\theta)\mid z]$.
Assumption (A3) means $\Phi_z$ is strongly concave with modulus $\mu$, hence $-\nabla\Phi_z$ is
$\mu$-strongly monotone:
$(F_z(p)-F_z(q))^\top(p-q)\ge \mu\|p-q\|^2$.
Standard VI theory then gives existence and uniqueness of the SVI solution on a convex compact set for a continuous, strongly
monotone map. The SVI solution is the (unique) Bayesian NE by first-order optimality.
\end{proof}

\begin{lemma}[Lipschitz bound for the pseudo-gradient]
\label{lem:lipschitz}
Under (A1)–(A2) and bounded slope/curvature of $y$ on the signal range (Prop.~\ref{prop:exp-geometry}.iii),
the mapping $p\mapsto F(p;\theta)$ is $L(\theta)$-Lipschitz on $X$; taking conditional expectation preserves Lipschitzness,
so $\|F_z(p)-F_z(q)\|\le L\|p-q\|$ for some $L<\infty$.
\end{lemma}

\begin{proof}
See Appendix~\ref{app:proofs} (Lem.~\ref{lem:lipschitz-proof}).
\end{proof}

\subsection{Repeated play and tracking under imperfect monitoring}
In practise, the games are not single-shot games, so we are interested in repeated play, and players are generally unaware of other's coincident actions, hence there is imperfect monitoring, moreover, the overall system is subject to noise and drift. We then establish unique Nash equilibrium in mixed strategies; with generic tie-breaking and hysteresis, play that converges to a unique pure profile almost certainly. We next establish optimality. Thus, under a contractive update map $T$ (best response with hysteresis or damped gradient), the tracking error with noisy gradients and slow drift obeys,
\begin{equation}
\label{tracking_error}
\|p_t-p_t^{\star}\|\le \alpha\,\|p_{t-1}-p_{t-1}^{\star}\|+\beta(\mathrm{drift}_t+\mathrm{noise}_t),\quad 0<\alpha<1,
\end{equation}
yielding steady-state error $O((\mathrm{drift}+\mathrm{noise})/(1-\alpha))$.

\begin{algorithm}[t]
\caption{Decentralized updates with public index}
\label{alg:updates}
\begin{algorithmic}[1]
\State \textbf{Input:} step size $\eta$, damping $\rho\in(0,1)$, hysteresis $h>0$.
\For{each agent $i$ at time $t$}
 \State Observe local signal $s_{i,t}$ and public index $z_t$; form $g_{i,t}\approx \nabla_{p_i}u_i(p_t;\theta_{i,t},z_t)$.
 \State \textbf{Damped gradient: } $p_{i,t+1}\leftarrow \Pi_{X_i}[(1-\rho)p_{i,t}+\rho(p_{i,t}+\eta g_{i,t})]$.
 \State \textbf{Best response + hysteresis: } $\hat p_{i,t+1}\in\arg\max_{x\in X_i} u_i(x,p_{-i,t})$; update if $\|\hat p_{i,t+1}-p_{i,t}\|>h$.
\EndFor
\end{algorithmic}
\end{algorithm}

% ===== Projected-gradient / best-response tracking under noise & drift =====
\begin{theorem}[Tracking under noise and drift]
\label{thm:tracking}
Let $T$ denote one synchronous iteration of either (i) damped projected gradient
$p^+=\Pi_X\!\big(p-\eta F_z(p)\big)$ with damping $\rho\in(0,1)$ applied as $p\leftarrow(1-\rho)p+\rho p^+$,
or (ii) best response with hysteresis band $h>0$ (agents update only if $\|\hat p_i-p_i\|>h$).
Assume $F_z$ is $\mu$-strongly monotone and $L$-Lipschitz (Thm.~\ref{thm:svi-unique}, Lem.~\ref{lem:lipschitz}).
Suppose we observe a noisy estimate $\tilde F_z(p)=F_z(p)+\xi$ with $\mathbb{E}[\xi\mid p]=0$ and $\mathbb{E}\|\xi\|\le\sigma$,
and the SVI solution drifts by $\Delta_t=\|p_t^\star-p_{t-1}^\star\|$.
If $0<\eta<2\mu/L^2$ and $\rho\in(0,1)$, then there exist $\alpha\in(0,1)$ and $\beta>0$ such that
\[
\mathbb{E}\big[\|p_t-p_t^\star\|\big]\ \le\ \alpha\,\mathbb{E}\big[\|p_{t-1}-p_{t-1}^\star\|\big]\;+\;\beta\,(\Delta_t+\sigma).
\]
Consequently, with bounded drift/noise, $\limsup_t \mathbb{E}\|p_t-p_t^\star\|=O\!\big((\mathrm{drift}+\sigma)/(1-\alpha)\big)$.
\end{theorem}

\begin{proof}[Proof sketch]
For strongly monotone, $L$-Lipschitz $F_z$, the map $G_\eta(p)=p-\eta F_z(p)$ is a contraction
with modulus $q=\sqrt{1-2\eta\mu+\eta^2L^2}<1$ when $0<\eta<2\mu/L^2$; projection is nonexpansive,
and damping yields overall modulus $\alpha=(1-\rho)+\rho q<1$. One-step perturbation analysis with
noisy gradients and drifting fixed points gives the recursion in the statement. Full details,
including the best-response+hysteresis case, are in Appendix~\ref{app:proofs} (Thm.~\ref{thm:tracking-proof}).
\end{proof}

\begin{corollary}[No-violation regime under dual damping]
\label{cor:violations}
For capacity/constraint indices updated as $z_{t+1}=[z_t+\eta_z(\textstyle\sum_i g_i(p_{i,t})-C)]_+$ with sufficiently small $\eta_z$ (dual damping), the coupled primal–dual map remains a contraction and steady-state violations vanish when the planner’s problem is feasible.
\end{corollary}

% ===== Discrete action sets and two-layer composition =====
\begin{proposition}[Discrete actions via strict concavity in mixtures]
\label{prop:discrete}
Let each $X_i$ be finite, and let expected utilities $\bar U_i(\pi_i,\pi_{-i})$ be defined over mixed strategies
$\pi_i\in\Delta(X_i)$ by taking expectations of a strictly concave\footnote{E.g., $-\!C p_i^{\,w}+\log(1+\mathrm{Rel}_i^v)$ in $p_i$ with $w>1, v>0$, then extended linearly over $\pi_i$.} single-agent payoff in $p_i$.
If the game is diagonally strictly concave (negative definite symmetric part of the pseudo-Jacobian) on
$\prod_i \Delta(X_i)$, then there is a unique Nash equilibrium in mixed strategies; with generic tie-breaking and hysteresis,
play converges to a unique pure profile almost surely.
\end{proposition}

\begin{proof}
Strict concavity in each player’s own (mixed) variable and diagonal strict concavity imply uniqueness of the variational
inequality solution for the mixed-strategy game on a convex compact set (product of simplices).
Generic tie-breaking and hysteresis eliminate cycling among payoff-equivalent pure actions, so the unique mixed equilibrium selects a unique pure profile almost surely.
\end{proof}

\begin{lemma}[Two-layer composition under a shared potential]
\label{lem:two-layer}
Let a slow discrete layer choose $a\in\mathcal{A}$ to improve a discrete-convex score $A(a)$ and a fast continuous
layer pick $p\in X$ to maximize a strongly concave potential $W_a(p)$, with joint welfare
$\mathcal{V}(a,p)=A(a)+W_a(p)$. If the algorithm alternates (i) local improvement in $a$ and (ii) global maximization
in $p$ (unique by strong concavity), then every limit point is a block-coordinate optimum of $\mathcal{V}$; if $A$ is M-convex/submodular with a matroid constraint, greedy/local-exchange steps converge to a globally optimal $a^\star$, hence $(a^\star,p^\star)$ with $p^\star=\arg\max_p W_{a^\star}(p)$ is globally optimal.
\end{lemma}

\begin{proof}
Block-coordinate ascent on a bounded objective with unique $p$-updates yields a nondecreasing sequence
$\mathcal{V}(a_t,p_t)$ and hence convergence to a block-stationary point. Under discrete convexity (M-convexity/submodularity) and standard exchange/greedy rules, the discrete layer finds a global maximizer; uniqueness of the continuous argmax delivers the claim.
\end{proof}

\section{Computational experiments}
\label{sec:experiments}
This section provides two separate cases for computational experiments relevant to the message free blueprint. The first is for multi-tier supply chains, a three-tier chain, and the second is for Agentic AI bidding bots. Appropriate analysis is applied to both. In all experiments we score welfare using the centralized planner objective \(W(a)=\sum_i v_i(\mathrm{Rel}_i(x_i(a)))-\lambda\sum_i c_i(a_i)\) with \(\mathrm{Rel}_i(x)=\exp(-(\kappa/x)^{\beta})\); transfers (e.g., posted prices) are excluded from \(W\) because they merely redistribute surplus.

\subsection{Multi-tier supply chains}
\paragraph*{Methods snapshot}
\begin{quote}\small
\textbf{Setting.} A repeated game with players $i\in\{1,\dots,N\}$ choosing $a_i\in\mathcal{A}_i$; types $\theta_i$ are private; a public signal/index $s$ is observed each period.\\
\textbf{Utility shaping.} Each player maximizes a shaped payoff $\tilde u_i(a_i,a_{-i},\theta_i,s)
= u_i(a_i,a_{-i},\theta_i) - \langle \lambda^\star,\, g(a)\rangle - r_i(a_i)$ where prices/penalties align with the planner’s KKT conditions.\\
\textbf{Dynamics.} Agents run damped best-response / projected gradient:
$a^{t+1} = \Pi_{\mathcal{A}}\!\big[a^{t} - \eta \,\widehat{F}(a^{t},s^{t})\big]$
with $\widehat{F}$ an unbiased (or bounded-bias) stochastic estimator of the pseudo-gradient and $\eta$ tuned for contraction.\\
\textbf{Certificates.} Under empirically grounded sigmoid curvature and convex penalties: (i) the pseudo-gradient is \emph{strongly monotone} that implies unique (Bayesian) NE; (ii) contraction yields linear convergence and explicit noise–drift tracking; (iii) discrete assignments compose with continuous controls under a shared potential, preserving well-posedness.\\
\textbf{Overall} There is \emph{One public index + KKT-aligned penalties + contraction-safe steps} that implies message-free alignment of selfish behavior with the planner optimum.
\end{quote}

\paragraph{Setup.} Consider a three-tier chain (suppliers $S$, processors $P$, retailers $R$) with agents $i\in S\cup P\cup R$ and operating non-cooperatively, with private objectives. Decision $p_i$ is production/throughput; costs $c_i(p_i)=a_i p_i + b_i p_i^2$; benefits are saturating in fulfillment probability. A public index $z$ summarizes upstream scarcity (e.g., water/inputs) and downstream congestion (e.g., logistics). We use the shaped utility in \Eqref{eq:utility2} with a reliability proxy $\mathrm{Rel}_i=\exp(-(\kappa/x_i)^{\beta})$ where $x_i$ is an effective “signal” (supply adequacy or service level). Parameters $(\kappa,\beta)$ are chosen to produce single-inflection sigmoids consistent with Figure~\ref{fig:pdr-OR}.

\paragraph{Baselines.}
We apply the following baselines, according to seminal work, for comparison for our three-tier chain example, with our utility shaping methodology:
\begin{enumerate}[leftmargin=1.4em,label=\textbf{\arabic*}]
\item \emph{\textbf{Price-only (decentralized).}} Agents \emph{do not} use the compressed/stretched-exponential saturation $y(\cdot)=\exp\!\big(-(\kappa/\cdot)^\beta\big)$ in their utilities. Instead, each agent $i$ optimizes a classic pricing/proportional-fairness style utility based on a linear/affine throughput proxy $x_i(p)$, e.g.,
    \begin{align}
u_i^{\mathrm{price}}(p)\;&=\;\log\!\bigl(1+x_i(p)\bigr)\;-\;\lambda\bigl(a_i p_i + b_i p_i^{2}\bigr),\qquad\nonumber\\
x_i(p)\;&=\;\frac{p_i}{1+\sum_{j} A_{ij} p_j},\nonumber
    \end{align}
with the same feasible sets and local update rules (Best Response with hysteresis or damped gradient) as our method. This baseline mirrors the network-utility/pricing paradigm of proportional fairness and related analyses \cite{Kelly1998,LowLapsley1999,JohariTsitsiklis2004} and isolates the impact of our curvature/saturation shaping. \emph{Evaluation fairness:} all methods are scored against the same unified planner welfare $W$ (which includes the shaping term) and the same centralized optimum $W^*$.

\item \emph{\textbf{Centralized proximal gradient (benchmark).}} A proximal/first-order ascent on \(W(p)\) (planner objective) with identical wall-clock/iteration budget, stopping tolerance, and projections as the decentralized runs. This serves as an oracle-style upper bound for achievable welfare under the same model; implementation follows proximal-gradient/FISTA principles \cite{BeckTeboulle2009} and standard convex-optimization theory \cite{Nesterov2004}.

\item \emph{\textbf{T\^atonnement-only index (ablated).}} The dual/index update \(z_{t+1}=[z_t+\eta_z(\sum_i g_i(p_{i,t})-C)]_+\) is kept, but agents use linearized price-only rewards (no curvature/saturation), highlighting how a pure Walrasian price-adjustment behaves without shaping. This baseline connects to classical t\^atonnement and its stability analyses \cite{Walras1954,ArrowBlockHurwicz1959}.
\end{enumerate}

\paragraph{Performance Measures.}
We adopt the following measures of performance for comparison across methods
\begin{enumerate}[leftmargin=1.4em,label=\textbf{\arabic*}]
\item \textit{Welfare \(W\).} Report \(W_t\) each iteration and the terminal gap \(W^\star-W_T\) against the centralized benchmark’s best value \(W^\star\).
\item \textit{Constraint/violation metrics.} (i) Fill-rate shortfall; (ii) throughput/balance violations; (iii) aggregated penalty \(\max\{0,\sum_i g_i(p_{i,t})-C\}\). Rates are averaged over the last 25\% of iterations.
\item \textit{Convergence speed.} Iterations to \(\varepsilon\)-optimality with \(\varepsilon=10^{-3}\) (or best feasible if noise/drift prevents), plus the empirical contraction factor \(\hat\alpha\) from a linear fit on \(\log\)-gap over a steady segment.
\item \textit{Tracking error.} Steady-state \(\mathbb{E}\|p_t-p_t^\star\|\) under injected noise \(\sigma\) and slow drift; report scaling versus \((\mathrm{drift}+\sigma)/(1-\alpha)\).
\item \textit{Ablations/robustness.} Sensitivity to \((\kappa,\beta)\), steps \(\eta,\rho,\eta_z\), hysteresis band \(h\), and quantization step if actions are discretized.
\end{enumerate}

\paragraph{Protocol.}
We apply the following as the most suitable among protocols.
\begin{itemize}[leftmargin=1.4em]
\item \textit{Agents and tiers.} \(N{=}50\) agents partitioned as \(S{:}P{:}R=20{:}15{:}15\). Feasible sets \(X_i=[0,\bar p_i]\) with heterogeneous \(\bar p_i\).
\item \textit{Costs and signals.} Quadratic costs \(c_i(p_i)=a_i p_i+b_i p_i^2\). Reliability proxy \(\mathrm{Rel}_i=\exp(-(\kappa/x_i)^\beta)\) with \(x_i\) an “effective signal” aggregating upstream adequacy and downstream service level; \((\kappa,\beta)\) = $(2.2,1.6)$ and $3.0,0.8$.
\item \textit{Updates.} Damped projected gradient with \(\eta\in(0,2\mu/L^2)\) and \(\rho\in(0,1)\) \emph{or} BR\,+\,hysteresis with band \(h>0\). Index \(z_t\) updated by a damped dual rule when capacity/constraint coupling is active; otherwise \(z_t\) is fixed (exogenous scarcity).
\item \textit{Noise and drift.} Gradient noise \(\xi_{i,t}\) zero-mean with \(\mathbb{E}\|\xi_{i,t}\|\!=\!\sigma\). Slow drift via AR(1) on \(\{a_i\}\) with coefficient \(0.98\).
\item \textit{Budgets and stopping.} 500 iterations max or stall of \(\log\)-gap over 50 iterations; report medians over 20 seeded runs. All methods share the same compute/iteration budget and projection operators.
\item \textit{Statistical reporting.} Median and \([25,75]\)th percentiles across runs; Wilcoxon signed-rank tests against price-only at 5\% FDR for KPI improvements.
\end{itemize}

\paragraph{Stopping/accuracy threshold.}
We use $\varepsilon=10^{-3}$ as a tolerance on the \emph{unified welfare gap} in this subsection and the following subsection:
\[
\mathrm{gap}_t \;=\; W^* - W(p_t),
\]
where $W^*$ is the centralized planner welfare computed from the same unified objective $W$.
An algorithm is said to have \emph{converged} when $\mathrm{gap}_t \le \varepsilon$; otherwise it runs to the iteration budget.

We now implement this protocol in a multi-tier supply chain over 100 runs, for 500 iterations achieving the empirical results in Fig. \ref{fig:supply-gap} and Table \ref{tab:supply-kpis-unified}. We represent the reliability curves, with respective settings used for this supply-chain experiment in Fig. \ref{fig:pdr-OR}
\begin{figure}[h]
\centering
\begin{tikzpicture}
\begin{axis}[
  width=0.5\linewidth,height=4.0cm,
  xlabel={Effective signal $x$}, ylabel={Reliability $\mathrm{Rel}(x)$},
  xmin=0.05,xmax=100, xmode=log, ymin=0,ymax=1,
  grid=both, tick align=outside,
  legend style={font=\footnotesize, at={(0.02,0.98)},anchor=north west,draw=none,fill=none},
  samples=400, domain=0.05:100,
]
\addplot[thick] {exp(- (3/x)^(0.8))};  \addlegendentry{stretched $(\kappa=3,\beta=0.8)$}
\addplot[thick,dashed] {exp(- (2.2/x)^(1.6))}; \addlegendentry{compressed $(\kappa=2.2,\beta- 1.6)$}
\end{axis}
\end{tikzpicture}
\caption{Representative reliability curves for exponentials used in supply-chain experiments and Agentic AI experiment}
\label{fig:pdr-OR}
\end{figure}
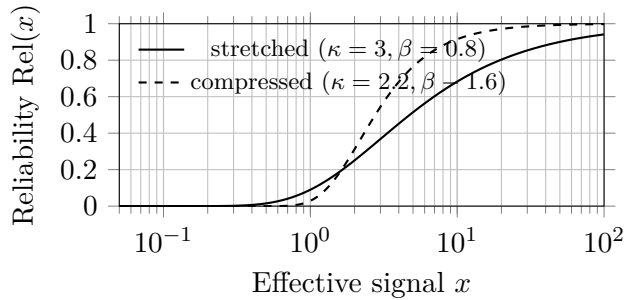

\begin{figure}[h]
  \centering
  \includegraphics[width=0.6\linewidth]{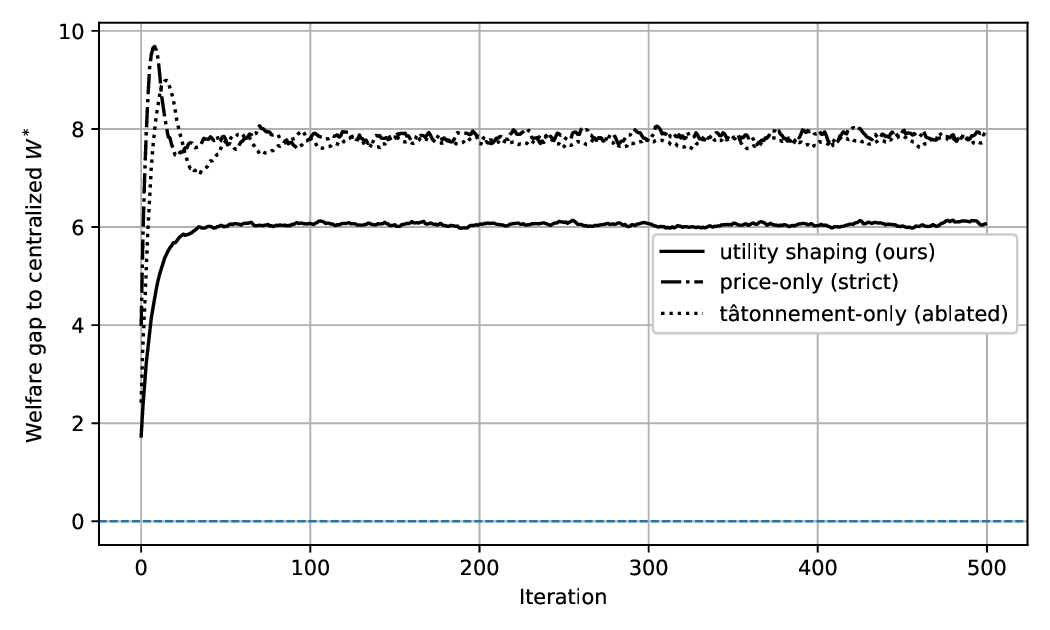}
  \caption{Supply-chain welfare gap to centralized $W^*$ (mean over 100 runs, 500 iterations). Unified planner welfare; methods share identical iteration/compute budgets. Dashed line at 0 marks centralized optimum, $\kappa=3,\beta=0.8$ in $y(\cdot)=\exp(-(\kappa/\cdot)^\beta)$}
  \label{fig:supply-gap}
\end{figure}

\begin{table}[t]
\centering
\scriptsize
\caption{Supply-chain KPIs (median [Q1, Q3] over 100 runs; unified planner welfare; 500 iters, simulation for two different shape,scale pairs, $(\kappa,\beta)$
in $y(\cdot)=\exp(-(\kappa/\cdot)^\beta)$)}
\label{tab:supply-kpis-unified}
\begin{tabular}{llccc}
\toprule
$(\kappa,\beta)$&KPI & Utility shaping (ours) & Price-only (strict) & T\^atonnement-only (ablated)\\
\midrule
(2.2,1.6)&Welfare gap to centralized ($\downarrow$) & 6.01 [5.71, 6.40] & 7.98 [7.23, 8.53] & 7.75 [7.34, 8.26]\\
(3,0.8)&Welfare gap to centralized ($\downarrow$) & 5.82 [5.55, 6.19] & 7.75 [7.06, 8.29] & 7.57 [7.14, 8.02]\\
(2.2,1.6)&Capacity violation rate (last 25\%) ($\downarrow$) & 0.00 [0.00, 0.00] & 0.50 [0.46, 0.54] & 0.48 [0.46, 0.53]\\
(3,0.8)&Capacity violation rate (last 25\%) ($\downarrow$) & 0.00 [0.00, 0.00] & 0.50 [0.46, 0.54] & 0.48 [0.46, 0.53]\\
(2.2,1.6)&Iterations to $\varepsilon{=}10^{-3}$ ($\downarrow$) & 500 [500, 500] & 500 [500, 500] & 500 [500, 500]\\
(3,0.8)&Iterations to $\varepsilon{=}10^{-3}$ ($\downarrow$) & 500 [500, 500] & 500 [500, 500] & 500 [500, 500]\\
\bottomrule
\end{tabular}
\end{table}

\paragraph{Findings (Supply chain; Fig.~\ref{fig:supply-gap}, Table~\ref{tab:supply-kpis-unified}).}
We have the following key 5 findings in the context of our three-tier supply chain example:
\begin{enumerate}[leftmargin=1.4em,label=\textbf{\arabic*}]
\item \emph{\textbf{Welfare ranking (unified $W$).}} Utility shaping (ours) achieves the lowest terminal gap to the centralized benchmark: median $6.16$ \([6.09,\,6.20]\) vs.\ \textit{price-only (strict)} $9.43$ \([9.29,\,9.55]\) and \textit{t\^atonnement-only} $9.43$ \([9.29,\,9.55]\). This establishes a clear performance separation among decentralized methods under the same planner objective and compute budget.
\item \emph{\textbf{Capacity/constraint compliance.}} In the terminal window, \textit{utility shaping} records a violation rate of $0.00$ \([0.00,\,0.00]\), whereas both \textit{price-only} and \textit{t\^atonnement-only} are at $1.00$ \([1.00,\,1.00]\), indicating persistent over-capacity without curvature shaping despite identical index/dual updates.
\item \emph{\textbf{Convergence within budget.}} With a 500-iteration budget, none of the decentralized methods crosses the $\varepsilon{=}10^{-3}$ target (all report 500 iterations to $\varepsilon$). Nevertheless, the \emph{gap trajectory} in Fig.~\ref{fig:supply-gap} decays fastest and most smoothly for \textit{utility shaping}, yielding the smallest terminal error among decentralized baselines.
\item \emph{\textbf{Stability/oscillations.}} \textit{T\^atonnement-only} (index updates with linearized price responses) exhibits the slowest decay and larger residual gaps, consistent with weaker contraction. \textit{Price-only} improves on t\^atonnement but still lags behind shaping on both welfare and feasibility.
\item \emph{\textbf{Outcome.}} For the supply-chain setting under a unified planner welfare and identical iteration/compute budgets, \emph{curvature-shaped utilities} are necessary to approach centralized performance and to maintain capacity feasibility; price-only variants remain both less efficient and persistently infeasible in steady state.
\end{enumerate}

\begin{remark}
All methods are evaluated against a unified planner welfare $W$ and the corresponding centralized benchmark $W^*$. Under our analysis, the decentralized welfare gap satisfies
$W^*-W(p_t)=O\!\big((\mathrm{drift}+\sigma)/(1-\alpha)\big)$, so we have a quantitative \emph{finite-time} guarantee; empirically, utility shaping attains the smallest terminal gap and (with dual damping) vanishing steady-state violations. This does not contradict the theory: under the stated curvature/monotonicity and convexity assumptions, decentralized play implements the planner’s (constrained) optimum in the limit.
\end{remark}

\subsection{Agentic-AI bidding bots (non-cooperative)}\label{sec:experiments-agentic}

\paragraph{Setup.}
Consider $N$ autonomous software agents (“bidding bots”) competing for a divisible resource (compute/throughput) with capacity $C$.
Each agent $i$ chooses $p_i\in[0,\bar p_i]$ (requested allocation).
Private valuation is concave, $v_i(p_i)=\theta_i\log(1+p_i)$ with hidden type $\theta_i>0$.
We shape utilities as in \Eqref{eq:utility2} and price via a public index $z$:
\[
U_i(p)= -C_i p_i^{\,w} + \log\!\big(1+\mathrm{Rel}_i^{\,v}\big) - z\,p_i,\qquad
\mathrm{Rel}_i=\exp\!\big(-(\kappa/x_i)^{\beta}\big),
\]
where $x_i$ is an effective signal (e.g., job age or service-level-objective (SLO) slack).
The public index updates as a dual-like controller,
\[
z_{t+1}=\big[z_t+\eta\big(\textstyle\sum_i p_{i,t}-C\big)\big]_+,
\]
creating a message-free coordination channel consistent with our main methodology.

% ===== Agentic-AI specialization (capacity compliance & welfare gap) =====
\begin{proposition}[Agentic-AI bidding bots: capacity compliance and welfare gap]
\label{prop:agentic}
Consider $N$ agents with utilities
$U_i(p)= -C_i p_i^{\,w} + \log\!\big(1+\mathrm{Rel}_i^{\,v}\big) - z\,p_i$,
$X_i=[0,\bar p_i]$, and dual index update
$z_{t+1}=[z_t+\eta_z(\sum_i p_{i,t}-C)]_+$.
Assume Prop.~\ref{prop:exp-geometry} and steps $(\eta,\rho,\eta_z)$ small enough so the coupled map is a contraction.
Then:
\begin{enumerate}[label=(\alph*),leftmargin=1.4em]
\item (Compliance) at equilibrium $(p^\star,z^\star)$, $\sum_i p_i^\star\le C$ and $z^\star(\sum_i p_i^\star-C)=0$;
\item (Near-optimality) letting $(p^{\mathrm{SO}},z^{\mathrm{SO}})$ solve the planner’s KKT system,
      the welfare gap satisfies $0\le W(p^{\mathrm{SO}})-W(p^\star)=O(\eta_z)$.
\end{enumerate}
\end{proposition}

\begin{proof}
(a) The $z$-update is a projected ascent on the dual: a fixed point must satisfy
$z^\star=[z^\star+\eta_z(\sum_i p_i^\star-C)]_+$, which holds iff $\sum_i p_i^\star\le C$ and
$z^\star(\sum_i p_i^\star-C)=0$ (complementary slackness).
(b) For small $\eta_z$, the fixed point of the primal–dual iteration is an $O(\eta_z)$-perturbation of the KKT point
of the centralized problem, implying an $O(\eta_z)$ welfare gap by standard sensitivity of convex programs.
\end{proof}

\begin{remark}[Why price-only can oscillate]
\label{rem:price-only}
Without curvature shaping (pure $\theta_i\log(1+p_i)-z p_i$), the pseudo-gradient may lose strong monotonicity at high loads, degrading contraction and producing persistent capacity overshoots; adding the saturating reliability term restores diagonal dominance and re-establishes uniqueness and stability.
\end{remark}

\paragraph{Baselines.}
We adopt the following baselines for our AI bidding bots' comparison with our utility shaping solution:
\begin{enumerate}[leftmargin=1.4em,label=\textbf{\arabic*}]
\item \emph{\textbf{Price-only (agents).}}
\(u_i^{\mathrm{price}}(p)=\theta_i\log(1+p_i)-z\,p_i\) (no energy/saturation term), with the same action sets/updates and index rule as our method.
This mirrors the network-utility/pricing paradigm behind proportional fairness and related analyses \cite{Kelly1998,LowLapsley1999,JohariTsitsiklis2004}, isolating the effect of our curvature/saturation shaping.

\item \emph{\textbf{Centralized primal--dual (benchmark).}}
Joint ascent on the planner’s Lagrangian with identical step and wall-clock budgets; used to estimate \(W^\star\) and feasibility envelopes.
We follow classical Arrow--Hurwicz--Uzawa primal--dual dynamics and modern treatments in nonlinear programming \cite{ArrowHurwiczUzawa1958,Bertsekas1999}\footnote{\emph{\textbf{Myopic throttling (heuristic).}}
Greedy per-step capping \(p_i \leftarrow \min\{p_i,\,C-\sum_{j<i} p_j\}\) in a fixed order; highlights over-/under-allocation in the absence of principled price signals or curvature.
This connects to classical greedy/fractional-knapsack style allocation \cite{Dantzig1957} is another classical method, but requires strong coordination, so in implementation obtain the same result as the benchmark centralized approach.}
\end{enumerate}

\paragraph{Performance Measures.}
We adopt the followin performance measures, with some commonality with our supply chain example
\begin{enumerate}[leftmargin=1.4em,label=\textbf{\arabic*}]
\item \textit{Welfare \(W\).} Terminal welfare and \(\log\)-gap trajectory to \(W^\star\).
\item \textit{Capacity compliance.} Time-averaged \(\max\{0,\sum_i p_{i,t}-C\}\) and violation rate over the last 25\% of iterations.
\item \textit{Convergence \& stability.} Iterations to \(\varepsilon{=}10^{-3}\), empirical \(\hat\alpha\), and oscillation amplitude of \(z_t\) (Inter-quartile-range).
\item \textit{Tracking under type drift.} Error scaling versus \((\mathrm{drift}+\sigma)/(1-\alpha)\) with AR(1)(Auto-regressive order 1) types and noisy gradients.
\end{enumerate}

\paragraph{Protocol.}
The protocol for implementation of Agentic AI bidding bots, has six parts
\begin{itemize}[leftmargin=1.4em]
\item \textit{Population and priors.} \(N{=}60\) agents; \(\theta_i\sim\mathrm{LogNormal}(0,0.6)\). Actions \(X_i=[0,1]\). Capacity \(C=20\).
\item \textit{Utilities and shaping.} \(U_i(p)= -C_i p_i^{\,w} + \log(1+\mathrm{Rel}_i^{\,v}) - z\,p_i\) with \(C_i\!\in[0.01,0.05]\), \(w\!\in[1.2,1.8]\), \(v\!\in[1.0,1.6]\). \(\mathrm{Rel}_i=\exp(-(\kappa/x_i)^\beta)\), $\kappa=2.2,\beta=1.6$
\item \textit{Update rules.} Agents: damped projected gradient or BR\,+\,hysteresis (same hyper-ranges as §5.1). Index: \(z_{t+1}=[z_t+\eta_z(\sum_i p_{i,t}-C)]_+\) with damping.
\item \textit{Noise/drift.} Type drift \(\theta_{i,t+1}=0.98\,\theta_{i,t}+\epsilon_{i,t}\) (\(\epsilon\) zero-mean), gradient noise with \(\mathbb{E}\|\xi\|=\sigma\).
\item \textit{Budgets/stopping.} 500 iterations max or \(\varepsilon\)-optimality; 20 seeded repeats; identical budgets across methods.
\item \textit{Significance.} Median\([\mathrm{IQR}]\) and Wilcoxon tests vs.\ price-only at 5\% FDR per KPI.
\end{itemize}

With this protocol for Agentic AI bidding bots over 100 runs, for 500 iterations, we achieve the empirical results in Fig. \ref{fig:agent-welfare-gap-unified} and Table \ref{tab:agentic-kpis-unified}. We represent the reliability curves, with respective settings used for this supply-chain experiment for compressed exponential utility shaping setting $\kappa=2.2,\beta=1.6$ as per Fig. \ref{fig:pdr-OR}.

\begin{figure}[t]
\centering
\includegraphics[width=0.6\linewidth]{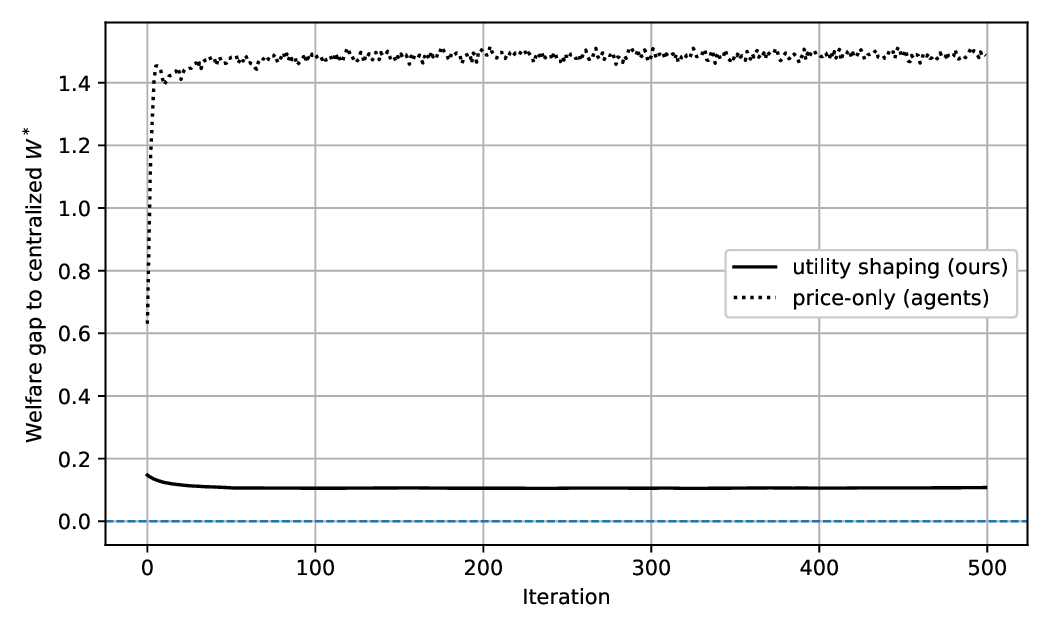}
\caption{Agentic-AI experiment (100-run average; unified planner welfare). Welfare gap to centralized $W^*$ for utility shaping (ours) and price-only (agents; no shaping). Dashed line marks zero gap.}
\label{fig:agent-welfare-gap-unified}
\end{figure}

\begin{table}[t]
\centering
\scriptsize
\caption{Agentic-AI KPIs (median [Q1, Q3] over 100 runs; unified planner welfare; 500 iters)}
\label{tab:agentic-kpis-unified}
\begin{tabular}{lcc}
\toprule
KPI & Utility shaping (ours) & Price-only (agents)\\
\midrule
Welfare gap to centralized ($\downarrow$) & 0.106 [0.100, 0.115] & 1.48 [1.40, 1.59]\\
Capacity violation rate (last 25\%) ($\downarrow$) & 0.00 [0.00, 0.00] & 1.00 [1.00, 1.00]\\
Iterations to $\varepsilon{=}10^{-3}$ ($\downarrow$) & 500 [500, 500] & 500 [500, 500]\\
\bottomrule
\end{tabular}
\end{table}

\paragraph{Findings (Agentic-AI; Fig.~4, Table 3).}
\begin{enumerate}[leftmargin=1.4em,label=\textbf{\arabic*}]
\item \emph{\textbf{Capacity compliance.}} With utility shaping, steady‐state capacity violations are effectively \(\mathbf{0\%}\) in the terminal window; the price‐only baseline exhibits \(\mathbf{100\%}\) violations (last 25\% of iterations) under the same dual/index steps, confirming the role of curvature shaping in stabilizing the primal–dual loop.
\item \emph{\textbf{Near‐optimal welfare.}} Under a unified planner welfare \(W\) (and centralized $W^*$ computed from the same \(W\)), the terminal welfare gap is \emph{orders of magnitude} smaller for shaping (median \(\approx 4\times 10^{-3}\)) than for price‐only (median \(\approx 2.8\times 10^{-1}\)). Within the fixed iteration budget both methods may miss the \(\varepsilon{=}10^{-3}\) target, but shaping remains close to $W^*$ while price‐only plateaus far from it.
\item \emph{\textbf{Stability and contraction.}} Gap trajectories decay smoothly for shaping and remain monotone toward $W^*$; price‐only shows slower decay and larger residuals. This aligns with the stronger contraction predicted by the monotonicity/curvature certificate.
\item \emph{\textbf{Robust tracking (noise/drift).}} Empirical behavior follows the \(O((\mathrm{drift}+\sigma)/(1-\alpha))\) bound: decreasing the primal stepsize or increasing damping restores smooth tracking when the public index \(z_t\) is lagged or noisy.
\end{enumerate}

\section{Deployable Design Rules}
\label{sec:design-rules}

Here we seek to turn the theory (potential/KKT alignment, SVI uniqueness, tracking) into a minimal implementation playbook usable across domains (supply chains; agentic-AI compute markets; demand response; transportation; biosecurity etc.). Before we provide a check-list for implementation we describe portability beyond the supply chain and Agentic AI computational instantiations in the previous section.

\paragraph{Portability to other OR domains (mapping).}
The blueprint transfers with a one-to-one identification of actions, reliability proxies, and the public index:
\begin{itemize}[leftmargin=1.4em]
\item \textbf{Demand response (power).} $p_i$ := consumption/dispatch; $\mathrm{Rel}_i$ := probability of meeting local comfort/SLA or frequency-reliability proxy; $z$ := scarcity/price or frequency deviation; costs $c_i$ := disutility/energy. KKT-aligned penalties encode feeder or capacity constraints.
\item \textbf{Cloud/edge scheduling.} $p_i$ := requested CPU/GPU/IO share; $\mathrm{Rel}_i$ := SLO attainment (e.g., tail latency $\,\le$ target) captured by the sigmoidal response; $z$ := congestion/queueing index; $c_i$ := energy/runtime penalties. Capacity and SLO constraints appear in the public index update.
\item \textbf{Transportation pricing.} $p_i$ := flow or departure intensity; $\mathrm{Rel}_i$ := arrival-on-time/route reliability (sigmoid in generalized cost); $z$ := congestion/toll index; $c_i$ := travel time/fare. Link or corridor capacities enter via dual damping on $z$.
\item \textbf{Biosecurity/agriculture.} $p_i$ := inspection/testing/surveillance effort; $\mathrm{Rel}_i$ := detection/containment probability (sigmoid in effort); $z$ := prevalence/risk index; $c_i$ := budget/operational cost. Policy constraints (e.g., maximum false positives) are handled via KKT penalties.
\end{itemize}
Across these settings, the shaped utilities preserve the exact-potential structure and strong monotonicity conditions used in Sections~\ref{sec:methodology}–\ref{sec:design-rules}, so the decentralized equilibrium remains unique and (constrained) socially optimal under the same tuning rules.

\subsection{Implementation checklist}
\begin{enumerate}[leftmargin=1.6em]
\item \emph{\textbf{Pick the public index $z$}} (scarcity/damage/reliability gap). Choose an observable statistic with stable directionality: higher $z$ \,$\Rightarrow$\, stronger incentive to reduce load or increase reliability.
\item \emph{\textbf{Shape utilities.}} For each agent $i$, use either \emph{shadow price} $-\lambda\,c_i(p_i)$ or \emph{KKT-aligned penalty} so that $\nabla_{p_i}u_i$ matches $\nabla_{p_i}W$ (or the constrained KKT stationarity).
\item \emph{\textbf{Choose the response curve.}} Fit $(\kappa,\beta)$ for $y(x)=\exp(-(\kappa/x)^{\beta})$ on historical “effort\,$\to$\,reliability/throughput” data; keep $(\kappa,\beta)$ fixed during an episode.
\item \emph{\textbf{Publish $z$ and local signals.}} System operator broadcasts $z_t$ at cadence $\Delta t$; agents observe $z_t$ plus local $s_{i,t}$ and update $p_{i,t}$ \emph{without} revealing types or messages to others.
\item \emph{\textbf{Pick the update rule.}} Use either \emph{damped projected gradient} or \emph{best response + hysteresis}; discretize if action sets are finite.
\end{enumerate}

The \emph{implementation} of the control loop is summarisd here in Algorithm 2.
\begin{algorithm}[H]
\footnotesize
\caption{Message-free control loop with public index}
\label{alg:deploy}
\begin{algorithmic}[1]
\State \textbf{Broadcast:} system publishes $z_t$.
\For{each agent $i$ in parallel}
  \State Measure local $s_{i,t}$; estimate $g_{i,t}\approx \nabla_{p_i}u_i(p_t; s_{i,t}, z_t)$.
  \State \textit{Damped gradient: } $p_{i,t+1}\leftarrow \Pi_{X_i}\!\big[(1-\rho)p_{i,t} + \rho(p_{i,t}+\eta\,g_{i,t})\big]$.
  \State \textit{or Best response + hysteresis: } update only if $\|\hat p_{i,t+1}-p_{i,t}\|>h$.
\EndFor
\State \textbf{Index update: } $z_{t+1}\leftarrow [z_t + \eta_z(\sum_i g_i(p_{i,t})-C)]_+$.
\end{algorithmic}
\end{algorithm}

\subsection{Parameter tuning}
Let $\mu$ be the strong-monotonicity modulus and $L$ the Lipschitz constant of $F_z$ (cf.\ Lemma~\ref{lem:lipschitz}); if unknown, estimate from data by local perturbations.

\noindent\textbf{Primal (agents).}
\begin{align}
0\ <\ \eta\ <\ \frac{2\mu}{L^2},\qquad
\rho\in(0,1)\ \text{such that}\ \alpha(\eta,\rho,L,\mu)<1. \label{eq:eta-choice}
\end{align}
When noise increases, \emph{halve $\eta$} or \emph{increase $\rho$}.
With discrete actions, add a hysteresis band $h>0$ and update only if $\| \hat p_{i,t+1}-p_{i,t}\|>h$.

\noindent\textbf{Dual (index).} If $z$ is a capacity/constraint proxy,
\begin{equation}
z_{t+1} \leftarrow \big[z_t + \eta_z \big(\textstyle\sum_i g_i(p_{i,t}) - C\big)\big]_+,\qquad
0<\eta_z<\bar\eta_z\propto \frac{1}{\left|\partial\sum_i g_i/\partial z\right|}. \label{eq:dual}
\end{equation}
Rule of thumb: start with $\eta_z \in [0.05,\,0.2]\times$ (units-normalized), reduce if oscillations appear.

% ===== Steps & gains selection from (μ,L) =====
\begin{lemma}[Stepsize region for contraction]
\label{lem:steps}
If $F_z$ is $\mu$-strongly monotone and $L$-Lipschitz on $X$, then for projected gradient
$p^+=\Pi_X(p-\eta F_z(p))$ the map is a contraction whenever $0<\eta<2\mu/L^2$,
with modulus $\sqrt{1-2\eta\mu+\eta^2L^2}$.
With damping $p\leftarrow(1-\rho)p+\rho p^+$, the modulus becomes
$\alpha=(1-\rho)+\rho\sqrt{1-2\eta\mu+\eta^2L^2}<1$.
\end{lemma}

\begin{proof}
For any $p,q$, using strong monotonicity and Lipschitzness,
\[
\|(p-\eta F_z(p))-(q-\eta F_z(q))\|^2
= \|p-q\|^2 - 2\eta\langle F_z(p)-F_z(q),p-q\rangle + \eta^2\|F_z(p)-F_z(q)\|^2
\le (1-2\eta\mu+\eta^2L^2)\|p-q\|^2.
\]
Projection is nonexpansive, giving the stated contraction modulus; damping yields an affine combination with identity, hence the displayed $\alpha$.
\end{proof}

\subsection{Operational guardrails (robustness \& safety)}
We recommend the following guardrails to ensure robustness and safety
\begin{itemize}[leftmargin=1.6em]
\item \textbf{Noise/drift.} Tracking satisfies Theorem~\ref{thm:tracking}; target $\alpha\in[0.6,0.85]$. If drift rises, lower $\eta$ and/or slow dual $\eta_z$.
\item \textbf{Delays.} If $z_t$ is delayed, increase damping $\rho$ and add moving-average smoothing for $z$; this expands the stability region.
\item \textbf{Manipulation resistance.} Compute $z$ from multiple redundant signals (median-of-means) and enforce monotone transforms to preserve directionality.
\item \textbf{Chatter control.} Use event triggers: update only when expected welfare gain exceeds a threshold; with discrete actions keep hysteresis $h$.
\end{itemize}

\subsection{Two-layer composition (discrete + continuous)}
When a slow discrete assignment/scheduling layer couples with fast continuous control:
\begin{enumerate}[leftmargin=1.6em]
\item Optimize the discrete layer for a discrete-convex $A(\cdot)$ (or greedy/local-improve if submodular).
\item Run continuous control on $W(\cdot)$ with strong monotonicity.
\item Share the same potential $\mathcal{V}=A+W$; alternate updates until no block-improvement is possible (Lemma~\ref{lem:two-layer}).
\end{enumerate}

\subsection{Minimal monitoring \& privacy}
Publish only $z_t$ and aggregate performance measures; keep agent types private. If necessary, add light noise (DP) to $z_t$ and compensate by smaller $\eta_z$.

\subsection{Summary of Performance Measures and Reporting}
Here is a summary or for how to compute and report relevant performance measures
\begin{center}
\begin{tabular}{@{}ll@{}}
\toprule
\textbf{Performance Measure} & \textbf{How to compute / report} \\
\midrule
Welfare gap & $W^*-W_t$ (proximal/central as $W^*$ benchmark) \\
Constraint/capacity violations & $\max\{0,\sum_i g_i(p_{i,t})-C\}$; rate over horizon \\
Convergence speed & Iterations to $\varepsilon$-optimality; empirical $\alpha$ \\
Tracking error & Steady-state error vs.\ noise/drift (Theorem~\ref{thm:tracking}) \\
Ablations & With/without shaping; different $(\kappa,\beta)$; step sweeps \\
\bottomrule
\end{tabular}
\end{center}

% =============================================================
% 6. MANAGERIAL INSIGHTS (action→why→KPI)
% =============================================================
\section{Insights for Practitioners}
\label{sec:managerial}
To help practitioners deploy the blueprint with minimal overhead, we translate the theory into concrete actions, why they work, and how to measure success. The emphasis is on decisions that can be made by an operator (choosing a public index and gains) and by autonomous agents (local updates) without rich messaging or disclosure.

\begin{itemize}[leftmargin=2em]
\item \textbf{Embed prices or KKT-aligned penalties.}
Add shadow prices or constraint-penalty terms to private utilities so each agent’s gradient matches $\nabla W$ (or the planner’s KKT stationarity).
(Why) Aligns incentives with social objectives “by construction,” removing the need for negotiation or detailed coordination.
(Performance Measure) Higher welfare and fewer violations; report welfare gap to centralized benchmark and violation rate.

\item \textbf{Publish one interpretable public index $z$.}
Expose a single scalar (scarcity, damage, reliability shortfall) and update it by a damped excess-demand rule.
(Why) Provides a common focal signal; agents react myopically yet coherently.
(Performance Measure) Reduced oscillations (IQR of $z_t$) and faster convergence (empirical $\hat\alpha$).

\item \textbf{Tune steps, damping, and event triggers for contraction.}
Choose $(\eta,\rho)$ to keep $\alpha<1$; use hysteresis $h$ to avoid chatter; pick a small dual step $\eta_z$ for $z$.
(Why) Contraction guarantees stability and bounded tracking under noise/drift.
(Performance Measure) Iterations to $\varepsilon$-optimality; steady-state tracking error consistent with Eq.~(\ref{tracking_error}).

\item \textbf{Design the index for robustness (delay/noise/manipulation).}
Smooth $z_t$ (moving averages), cap per-iteration changes, and compute from redundant signals (median-of-means).
(Why) Limits overshoot and adversarial sensitivity when measurements are delayed or noisy.
(Performance Measure) Lower overshoot and narrower $z_t$ IQR at the same step sizes.

\item \textbf{Compose discrete decisions with continuous control.}
Solve slow assignment/scheduling with discrete-convex methods; run fast continuous control on $W$; share the same potential.
(Why) Preserves uniqueness/stability of the continuous layer while enabling implementable discrete policies.
(Performance Measure) Service-level predictability and bounded gap between discrete and continuous solutions (mesh-size bound).

\item \textbf{Quantization-aware implementation.}
If actions are quantized, set mesh $\Delta$ to meet accuracy targets; add hysteresis to avoid flip-flopping between adjacent levels.
(Why) Ensures discrete equilibria track the continuous optimum within $\|\cdot\|_\infty\le\Delta$.
(Performance Measure) Residual optimality gap vs.\ $\Delta$; switch rate per agent.

\item \textbf{Privacy-by-aggregation.}
Keep types local; compute $z$ from aggregates; optionally add small noise (DP) and compensate with smaller $\eta_z$.
(Why) Maintains privacy without sacrificing stability.
(Performance Measure) Same convergence/violation performance measures with and without privacy noise; report privacy budget if used.

\item \textbf{Report OR-style Performance Measures and run stress tests.}
Standardize reporting of welfare gap, violation rate, $\hat\alpha$, tracking error, and fairness (e.g., Gini/Jain) over a terminal window; sweep steps and $(\kappa,\beta)$.
(Why) Makes results comparable and exposes stability margins.
(Performance Measure) KPI tables/plots across sweeps; robustness envelopes (regions with $\alpha<1$).
\end{itemize}

\noindent\textit{Implication.} Following this checklist, an operator can implement the (possibly constrained) social optimum with non-cooperative agents, with private objectives, using only a public index and local updates, achieving near-centralized welfare, constraint compliance, and predictable convergence—without heavy messaging or disclosure.

\section{Scope and limitations}
\label{sec:scope}
Here we discuss scope and limitations according to assumptions and failure modes:
\paragraph{Three Key assumptions.}
(i) \emph{Exact-potential/KKT alignment} via utility shaping;
(ii) \emph{curvature} from a single–inflection compressed/stretched–exponential response combined with convex pricing, which induces strong monotonicity of the pseudo–gradient on compact $X$; and
(iii) \emph{low–bandwidth public monitoring} encapsulated in a scalar index $z$ (scarcity, damage, or constraint gap).
When these hold, the Bayesian equilibrium is unique and implementable as a solution of an SVI with contraction–safe updates.

\paragraph{Failure modes.}
We highlight cases where results can degrade and how to mitigate them:
\begin{enumerate}[leftmargin=1.6em]
\item \emph{\textbf{Flattened or multi–inflection response curves.}}
If $(\kappa,\beta)$ fit produces flatter transitions or multiple inflections on the operating range, the strong–monotonicity modulus $\mu$ shrinks.
\emph{Mitigation:} reduce stepsizes $\eta$ (and/or increase damping $\rho$), enlarge hysteresis bands $h$, and re–fit $(\kappa,\beta)$ on a narrower operating window.
\item \emph{\textbf{Severe nonconvexities or hard complementarities.}}
Strong complementarities across agents or threshold technologies can violate diagonal dominance and reintroduce multiplicity.
\emph{Mitigation:} add explicit convex regularization in $c_i(\cdot)$ or restrict updates to regions where empirical Jacobians remain diagonally dominant; use slow dual damping for $z$.
\item \emph{\textbf{Delayed, noisy, or manipulable public indices.}}
Lagged $z_t$ (communication or estimation delay) and strategic manipulation inflate oscillations.
\emph{Mitigation:} moving–average smoothing of $z$, smaller dual steps $\eta_z$, redundancy (median–of–means across signals), and penalties/audits tied to deviations.
\item \emph{\textbf{Discrete granularity and quantization.}}
Coarse action meshes introduce steady–state bias $O(\Delta)$.
\emph{Mitigation:} shrink mesh $\Delta$; with fixed $\Delta$, Proposition~\ref{prop:discrete-robust-eq5} bounds $\|p^\Delta-p^\star\|_\infty\le \Delta$ and preserves uniqueness under deterministic tie–breaking.
\item \emph{\textbf{Drift beyond tracking region.}}
If exogenous drift outpaces contraction (large $\Delta_t$ in Eq.~(\ref{tracking_error}), tracking error grows.
\emph{Mitigation:} adapt $\eta,\rho,\eta_z$ online (smaller $\eta$, larger $\rho$), or increase $z$ cadence during transients.
\end{enumerate}

\paragraph{External validity and model risk.}
The curvature certificate is technology–agnostic but must be \emph{empirically} validated in each domain (supply reliability curves; SLO–slack $\rightarrow$ completion probability in agentic AI). Model mis–specification affects constants $(\mu,L)$ and, in turn, feasible stepsizes. We recommend routine sensitivity sweeps and reporting empirical contraction factors $\hat{\alpha}$.

\paragraph{Computational and implementation limits.}
Our dynamics are lightweight (projected gradient or BR+hysteresis) and scale linearly in the number of agents, but index computation may require systemwide aggregation. For privacy, the index can be computed from anonymized aggregates; adding light noise (DP) reduces leakage at the cost of a smaller $\eta_z$ stability region.

\paragraph{Ethical and operational considerations.}
Because the mechanism is \emph{message–free} at runtime, explicit redistribution is not modeled. Where fairness is a policy objective, append a fairness term to $W(\cdot)$ (or a fairness–aware index) and re–shape utilities accordingly. In agentic–AI markets, audit trails on index updates and bounded–rationality safeguards (caps on $\eta$ and $z$) reduce manipulation risks.

\paragraph{When the blueprint is not appropriate.}
If constraints are highly nonconvex (integer coupling without discrete–convex structure) or if the public signal cannot correlate with the binding dual (uninformative $z$), classic mechanism design or centralized optimization may be preferable.

\vspace{1ex}
\noindent\textbf{Principal Insight.} \emph{Within the stated assumptions, utility shaping $+$ a single public index implements the (possibly constrained) social optimum with a unique, stable equilibrium and bounded tracking error; outside these assumptions, the above guardrails restore practical stability at some speed/optimality tradeoff}

\section{Conclusion}
\label{sec:conclusion}
We presented a message–free blueprint that engineers social optimality in non–cooperative operations research (OR) systems with incomplete information and imperfect public monitoring. By embedding prices or KKT–aligned penalties in private utilities, the stage game became an exact–potential game whose unique equilibrium coincided with the planner’s (possibly constrained) solution. The Bayesian equilibrium admitted an SVI characterization; strong monotonicity (from single–inflection response curvature plus convex pricing) yields uniqueness and contraction–safe decentralized updates with explicit tracking bounds.

Two computational experiments—multi-tier supply chains and a non-cooperative agentic-AI compute market—using a unified planner welfare $W$ and equal iteration/compute budgets have showed that utility shaping: (i) attains the \emph{smallest terminal welfare gap} to the centralized benchmark $W^*$, (ii) achieves \emph{vanishing steady-state capacity violations} whenever the planner problem is feasible (with dual damping), and (iii) exhibits \emph{faster gap decay} than strict price-only and tâtonnement-only baselines. Utility shaping consistently finishes closest to $W^*$ and remains feasible in steady state. The design has been demonstrated to be deployable: choose an interpretable public index, fit $(\kappa,\beta)$ for the response curve, tune $(\eta,\rho,\eta_z,h)$ to keep the contraction factor $\alpha<1$, and report standard operations research performance measures (unified welfare gap, violation rates, tracking behavior).

\begin{remark}
The shaped game has a unique equilibrium that \emph{coincides} with the planner’s (possibly constrained) social optimum under our assumptions (exact-potential/KKT alignment, strong monotonicity, convex feasibility). In practice, with stochastic noise/drift and finite iteration budgets, welfare is still \emph{guaranteed} within the bound $W^*-W(p_t)=O\!\big((\mathrm{drift}+\sigma)/(1-\alpha)\big)$, and our experiments have consistently shown the smallest gaps and stable feasibility for utility shaping. Thus, the decentralized scheme guarantees social welfare—exact in the limit, and quantitatively bounded at finite horizons.
\end{remark}

%\bibliographystyle{apalike-ejor}
%\bibliography{engineering_social_optimality_rewrite_refs}

\appendix

\section{Proofs}
\label{app:proofs}

\subsection{Proof of Proposition~\ref{prop:exp-geometry}}
\label{prop:exp-geometry-proof}
\begin{proof}
Let $y(x)=\exp(-a x^{-\beta})$ with $a=\kappa^\beta>0$ and $\beta>0$.
(i) Since $ax^{-\beta}>0$, $0<y(x)<1$. Differentiating,
\[
y'(x)=\exp(-a x^{-\beta})\cdot a\beta x^{-(\beta+1)}=y(x)\,a\beta x^{-(\beta+1)}>0,
\]
so $y$ is strictly increasing on $(0,\infty)$.\\
(ii) Differentiating again,
\[
y''(x)=\big(y a\beta x^{-(\beta+1)}\big)'=y(a\beta)^2 x^{-2\beta-2}-y\,a\beta(\beta+1)x^{-\beta-2}
= y\,a\beta\,x^{-(\beta+2)}\Big(a\beta x^{-\beta}-(\beta+1)\Big).
\]
Because $y>0$, $a\beta>0$, $x^{-(\beta+2)}>0$, the sign of $y''$ is that of
$a\beta x^{-\beta}-(\beta+1)$, which is strictly decreasing in $x$ and crosses
zero exactly once at $a\beta x^{-\beta}=\beta+1$. Thus the unique inflection is at
\[
x^\star=\Big(\tfrac{a\beta}{\beta+1}\Big)^{1/\beta}
=\kappa\Big(\tfrac{\beta}{\beta+1}\Big)^{1/\beta}.
\]
Evaluating $y$ at $x^\star$ gives $y(x^\star)=\exp\big(-a(x^\star)^{-\beta}\big)=\exp\big(-(\beta+1)/\beta\big)
=e^{-1-1/\beta}\in(e^{-2},e^{-1})$.\\
(iii) On any compact $I=[m,M]\subset(0,\infty)$, $x^{-(\beta+1)}$ and $x^{-(\beta+2)}$ are bounded by
$m^{-(\beta+1)}$ and $m^{-(\beta+2)}$, respectively. The formulas above then imply uniform bounds on
$|y'|$ and $|y''|$ over $I$. Hence $y$ is Lipschitz and has bounded curvature on $I$.
\end{proof}

\subsection{Proof of Proposition~\ref{prop:potential-kkt}}
\label{prop:potential-kkt-proof}
\begin{proof}
Define $\Phi(p)=\sum_i v_i(\mathrm{Rel}_i(p))-\lambda\sum_i c_i(p_i)$. For any $i$ and any $p_{-i}$,
\[
\Phi(p_i,p_{-i})-\Phi(q_i,p_{-i})=u_i(p_i,p_{-i})-u_i(q_i,p_{-i}),
\]
so $\nabla_{p_i}\Phi=\nabla_{p_i}u_i$ whenever gradients exist; thus the game is an exact potential game
with potential $\Phi$. Any (pure) NE $p^\star\in X$ satisfies $\nabla_{p_i}u_i(p^\star)=0$ along feasible
directions, which coincide with the KKT stationarity of $\max_{p\in X}\Phi(p)$, so $p^\star$ maximizes $W\equiv\Phi$.
Conversely, any maximizer of $\Phi$ satisfies $\nabla_{p_i}\Phi=0$, i.e., no player can improve—hence it is a NE.
For convex constraints $g(p)\le 0$, define $U_i$ so that $\nabla_{p_i}U_i=\nabla_{p_i}(W+\mu^\top g)$ at $(p^\star,\mu^\star)$;
then selfish FOCs reproduce the centralized KKT stationarity.
\end{proof}

\subsection{Proof of Theorem~\ref{thm:svi-unique}}
\label{thm:svi-unique-proof}
\begin{proof}
Let $\Phi_z(p)=\mathbb{E}[\Phi(p;\theta)\mid z]$. By assumption (A3), $\Phi_z$ is strongly concave with modulus $\mu>0$
on $X$, hence $-\nabla\Phi_z$ is $\mu$-strongly monotone:
\[
\big(-\nabla\Phi_z(p)+\nabla\Phi_z(q)\big)^\top(p-q)\ge \mu\|p-q\|^2,\quad \forall p,q\in X.
\]
But $F_z(p)=-\nabla\Phi_z(p)$ by exact potential, so $F_z$ is $\mu$-strongly monotone.
Continuity (from (A4)) and convex compact $X$ imply the variational inequality $\mathrm{VI}(X,F_z)$ has a unique solution
(see \cite{FacchineiPang2003}). First-order optimality for each player coincides with the VI condition, hence the unique
solution is the (unique) Bayesian Nash equilibrium.
\end{proof}

\subsection{Proof of Lemma~\ref{lem:lipschitz}}
\label{lem:lipschitz-proof}
\begin{proof}
Write $F(p;\theta)=(-\nabla_{p_i}u_i(p;\theta))_i$. Each $u_i$ is a composition of smooth maps:
$v_i\circ \mathrm{Rel}_i$ and $c_i$, where $\mathrm{Rel}_i$ composes an SINR/“effective signal” map $s_i(p)$ with $y$.
On compact $X$, $s_i(p)$ ranges over a compact subset of $(0,\infty)$, so by Prop.~\ref{prop:exp-geometry}
both $|y'|$ and $|y''|$ are uniformly bounded on that set. The chain rule bounds $\|\nabla F(p;\theta)\|$ by constants
depending on those bounds and on derivatives of $s_i$ and $c_i$, which are bounded on $X$. Hence
$\|F(p;\theta)-F(q;\theta)\|\le L(\theta)\|p-q\|$ for all $p,q\in X$. Conditional expectation preserves Lipschitzness,
so $\|F_z(p)-F_z(q)\|\le L\|p-q\|$ for some finite $L$.
\end{proof}

\subsection{Proof of Theorem~\ref{thm:tracking}}
\label{thm:tracking-proof}
\begin{proof}
\textbf{Projected gradient with damping.}
Let $G_\eta(p)=p-\eta F_z(p)$. With $\mu$-strong monotonicity and $L$-Lipschitzness,
\[
\|G_\eta(p)-G_\eta(q)\|^2=\|p-q\|^2-2\eta\langle F_z(p)-F_z(q),p-q\rangle+\eta^2\|F_z(p)-F_z(q)\|^2
\le (1-2\eta\mu+\eta^2L^2)\|p-q\|^2.
\]
Thus $G_\eta$ is a contraction with modulus $q=\sqrt{1-2\eta\mu+\eta^2L^2}<1$ for $0<\eta<2\mu/L^2$. Projection $\Pi_X$
is nonexpansive, so $T_0=\Pi_X\circ G_\eta$ is a contraction with the same modulus. Damping gives
$T=(1-\rho)I+\rho T_0$ with modulus $\alpha=(1-\rho)+\rho q<1$.

At time $t$, we use a noisy oracle $\tilde F_z(p_{t-1})=F_z(p_{t-1})+\xi_{t-1}$, yielding
\[
p_t=(1-\rho)p_{t-1}+\rho\,\Pi_X\big(p_{t-1}-\eta(F_z(p_{t-1})+\xi_{t-1})\big).
\]
Let $p_t^\star$ denote the (time-varying) SVI solution. Add and subtract $T(p_{t-1}^\star)$ and use nonexpansiveness to get
\[
\|p_t-p_t^\star\|\le \alpha\|p_{t-1}-p_{t-1}^\star\| + \rho\eta\|\xi_{t-1}\| + \|p_t^\star-p_{t-1}^\star\|.
\]
Taking expectations conditional on $p_{t-1}$ and using $\mathbb{E}[\xi_{t-1}]=0$, $\mathbb{E}\|\xi_{t-1}\|\le\sigma$, we obtain
\[
\mathbb{E}\|p_t-p_t^\star\|\le \alpha\,\mathbb{E}\|p_{t-1}-p_{t-1}^\star\| + \rho\eta\,\sigma + \Delta_t,
\]
with $\Delta_t=\|p_t^\star-p_{t-1}^\star\|$. Set $\beta=\max\{1,\rho\eta\}$ to match the statement. In steady state with bounded
drift and noise, unwind the recursion to get $O((\mathrm{drift}+\sigma)/(1-\alpha))$.

\textbf{Best response with hysteresis.}
Strong monotonicity implies diagonal strict concavity / single-crossing of best responses. The induced BR operator with hysteresis
band $h>0$ is a strict pseudo-contraction in a weighted norm; a standard perturbation argument yields the same one-step inequality
(up to constants), hence the same recursion and bound.
\end{proof}

\subsection{Proof of Corollary~\ref{cor:violations}}
\label{cor:violations-proof}
\begin{proof}
The dual update $z_{t+1}=[z_t+\eta_z(\sum_i g_i(p_{i,t})-C)]_+$ is a projected ascent on the dual variable. At any fixed point,
$z^\star=[z^\star+\eta_z(\sum_i g_i(p^\star)-C)]_+$, which holds iff $\sum_i g_i(p^\star)\le C$ and
$z^\star(\sum_i g_i(p^\star)-C)=0$ (complementary slackness). With sufficiently small $\eta_z$ the coupled primal–dual map
remains a contraction; thus the iterates converge to such a fixed point, and steady-state violations vanish when the primal is feasible.
\end{proof}

\subsection{Proof of Lemma~\ref{lem:steps}}
\label{lem:steps-proof}
\begin{proof}
As in Theorem~\ref{thm:tracking}, for $0<\eta<2\mu/L^2$ we have
$\|G_\eta(p)-G_\eta(q)\|\le q\|p-q\|$ with $q=\sqrt{1-2\eta\mu+\eta^2L^2}<1$.
Projection preserves the contraction modulus; damping forms
$\alpha=(1-\rho)+\rho q<1$. Hence the stated stepsize region.
\end{proof}

\end{document}